\documentclass[]{bytedance_seed}

\usepackage[toc,page,header]{appendix}

\usepackage{minitoc}
\usepackage{longtable}
\newcommand{\sd}[1]{\ensuremath{\pm #1}}
\usepackage{wrapfig}
\usepackage{graphicx}
\usepackage{multirow}
\usepackage{multicol}
\usepackage{amsthm}

\newcommand{\Cell}[1]{\begin{tabular}{@{}c@{}}#1\end{tabular}}

\usepackage{float}

\usepackage{wrapfig}
\usepackage{caption}
\usepackage{algorithmic}
\usepackage{algorithm}
\usepackage{subcaption}
\usepackage[utf8]{inputenc} 
\usepackage[T1]{fontenc}    
\usepackage{url}            
\usepackage{booktabs}       
\usepackage{amsfonts}       
\usepackage{nicefrac}       
\usepackage{microtype}      
\usepackage{xcolor}         
\usepackage{amsmath}
\usepackage{amssymb}
\usepackage{mathtools}
\usepackage{slashed}
\usepackage{braket}

\usepackage{enumitem}

\usepackage{cleveref}
\usepackage[colorinlistoftodos,textsize=small]{todonotes}
\usepackage{graphicx}
\usepackage{tikz}
\usepackage{tikz-cd}

\usepackage{bm}
\usepackage{physics}
\usepackage{xcolor}
\usepackage{natbib}
\usepackage{mdframed}
\usepackage{nicefrac}
\usepackage{booktabs}
\usepackage{lipsum}
\usepackage{titlesec}
\usepackage{wrapfig,lipsum,booktabs}
\usepackage{blindtext}
\usepackage[most]{tcolorbox}

\definecolor{darkgreen}{rgb}{0,0.40,0}
\definecolor{firebrick}{rgb}{0.698,0.133,0.133}

\let\cite\citep

\definecolor{Blue}{rgb}{0, 0, 0.8}
\definecolor{blue}{rgb}{0,0,1}

\newtheoremstyle{theoremstyle}
  {.5\baselineskip} %
  {.5\baselineskip} %
  {}                  %
  {}                  %
  {\bfseries}        %
  {.}                 %
  {1em}               %
  {}                  %

\newtheoremstyle{claimstyle}
  {.5\baselineskip}   %
  {.5\baselineskip}   %
  {}                  %
  {}                  %
  {\bfseries}         %
  {.}                 %
  {1em}               %
  {\thmname{{\color{findingrule}#1}}\thmnumber{\ {\color{findingrule}#2}}\thmnote{ {\normalcolor(#3)}}}

\theoremstyle{theoremstyle}

\theoremstyle{claimstyle}
\newtheorem{lemma}{Lemma}[section]
\newtheorem{proposition}{Proposition}
\newtheorem{definition}{Definition}[section]
\theoremstyle{theoremstyle}

\tcolorboxenvironment{theorem}{
  breakable,
  colback=black!10,
  colframe=white,%
  width=\dimexpr\linewidth+10pt\relax,%
  enlarge left by=-5pt,%
  enlarge right by=-5pt,%
  boxsep=5pt,%
  boxrule=0pt,
  left=0pt,right=0pt,top=0pt,bottom=0pt,
  sharp corners,
  before skip=0.5\baselineskip, %
  after skip=0.5\baselineskip,  %
  fonttitle=\bfseries, %
  coltitle=black %
}
\tcolorboxenvironment{remark}{
  blanker,
  breakable,
  before skip=.8\baselineskip,  %
  after  skip=.8\baselineskip   %
}

\tcolorboxenvironment{proposition}{
  breakable,
  colback=findingbg,
  colframe=findingrule,
  boxrule=0.7pt,
  boxsep=6pt,
  left=2pt, right=2pt, top=2pt, bottom=2pt,
  sharp corners,
  fontupper=\itshape,
  fonttitle=\bfseries,
  coltitle=findingrule,
  before skip=0.85\baselineskip,
  after skip=0.85\baselineskip,
}

\tcolorboxenvironment{lemma}{
  breakable,
  colback=findingbg,
  colframe=findingrule,
  boxrule=0.7pt,
  boxsep=6pt,
  left=2pt, right=2pt, top=2pt, bottom=2pt,
  sharp corners,
  fontupper=\itshape,
  fonttitle=\bfseries,
  coltitle=findingrule,
  before skip=0.85\baselineskip,
  after skip=0.85\baselineskip,
}

\tcolorboxenvironment{corollary}{
  breakable,
  colback=black!10,
  colframe=white,%
  width=\dimexpr\linewidth+10pt\relax,%
  enlarge left by=-5pt,%
  enlarge right by=-5pt,%
  boxsep=5pt,%
  boxrule=0pt,
  left=0pt,right=0pt,top=0pt,bottom=0pt,
  sharp corners,
  before skip=0.5\baselineskip, %
  after skip=0.5\baselineskip,  %
  fonttitle=\bfseries, %
  coltitle=black %
}

\tcolorboxenvironment{definition}{
  breakable,
  colback=findingbg,
  colframe=findingrule,
  boxrule=0.7pt,
  boxsep=6pt,
  left=2pt, right=2pt, top=2pt, bottom=2pt,
  sharp corners,
  fontupper=\itshape,
  fonttitle=\bfseries,
  coltitle=findingrule,
  before skip=0.85\baselineskip,
  after skip=0.85\baselineskip,
}
\tcolorboxenvironment{assumption}{
  breakable,
  colback=black!10,
  colframe=white,%
  width=\dimexpr\linewidth+10pt\relax,%
  enlarge left by=-5pt,%
  enlarge right by=-5pt,%
  boxsep=5pt,%
  boxrule=0pt,
  left=0pt,right=0pt,top=0pt,bottom=0pt,
  sharp corners,
  before skip=0.5\baselineskip, %
  after skip=0.5\baselineskip,  %
  fonttitle=\bfseries, %
  coltitle=black %
}

\definecolor{ablnew}{HTML}{000000}   %
\definecolor{newedit}{HTML}{000000}  %
\definecolor{newedit2}{HTML}{000000} %
\definecolor{newedit3}{HTML}{000000} %
\definecolor{fixup}{HTML}{000000}    %
\definecolor{findingbg}{HTML}{EAF3FA}
\definecolor{findingrule}{HTML}{0072B2}
\newcounter{finding}
\newcounter{hypothesis}

\newcommand{\claimbox}[2]{%
  \par\smallskip\noindent
  {\setlength{\fboxrule}{0.7pt}\setlength{\fboxsep}{6pt}%
   \fcolorbox{findingrule}{findingbg}{%
     \begin{minipage}{\dimexpr\linewidth-2\fboxsep-2\fboxrule\relax}
       {\color{findingrule}\bfseries #1.}~#2\par
     \end{minipage}}}%
  \par\smallskip}

\newcommand{\finding}[1]{\refstepcounter{finding}\claimbox{Finding \thefinding}{#1}}

\tcolorboxenvironment{fact}{
  breakable,
  colback=black!10,
  colframe=white,%
  width=\dimexpr\linewidth+10pt\relax,%
  enlarge left by=-5pt,%
  enlarge right by=-5pt,%
  boxsep=5pt,%
  boxrule=0pt,
  left=0pt,right=0pt,top=0pt,bottom=0pt,
  sharp corners,
  before skip=0.5\baselineskip, %
  after skip=0.5\baselineskip,  %
  fonttitle=\bfseries, %
  coltitle=black %
}
\usepackage{cleveref}
\crefname{theorem}{Theorem}{Theorems}
\crefname{equation}{Equation}{Equation}
\crefname{algorithm}{Algorithm}{Algorithm}
\crefname{appendix}{Appendix}{Appendices}
\crefname{proposition}{Proposition}{Propositions}
\crefname{lemma}{Lemma}{Lemmas}
\crefname{corollary}{Corollary}{Corollaries}
\crefname{definition}{Definition}{Definitions}
\crefname{assumption}{Assumption}{Assumptions}
\crefname{remark}{Remark}{Remarks}
\crefname{problem}{Problem}{Problems}
\crefname{property}{Property}{property}
\crefname{section}{Section}{Sections}
\crefname{appendix}{Appendix}{Appendices}
\crefname{proof}{Proof}{Proofs}
\crefname{figure}{Figure}{Figures}
\crefname{table}{Table}{Tables}

\numberwithin{equation}{section}
\numberwithin{theorem}{section}
\numberwithin{definition}{section}
\numberwithin{lemma}{section}
\numberwithin{assumption}{section}
\numberwithin{remark}{section}

\usepackage{makecell}

\titlespacing\section{0pt}{4pt plus 4pt minus 2pt}{-2pt plus 2pt minus 2pt}
\titlespacing\subsection{0pt}{2pt plus 4pt minus 2pt}{-2pt plus 2pt minus 2pt}
\titlespacing\subsubsection{0pt}{2pt plus 4pt minus 2pt}{-2pt plus 2pt minus 2pt}
\usepackage{caption}

\usepackage{titletoc}

\newcommand{\sys}{{\sc EMBER}\xspace}
\newcommand{\syb}{{\sc \textbf{EMBER}}\xspace}

\def\Snospace~{\S{}}

\newcommand{\sref}[2]{\hyperref[#2]{#1 \ref{#2}}}

\title{EMBER: Predicting TCR--pMHC Binding Without a Docked Complex
}

\usepackage{array}
\usepackage{colortbl}

\definecolor{LightCyan}{rgb}{0.8, 0.9, 1}
\newcolumntype{b}{>{\columncolor{LightCyan}\hspace{0pt}}c}
\definecolor{cadetgrey}{rgb}{0.57, 0.64, 0.69}
\newcolumntype{g}{>{\columncolor{cadetgrey}\hspace{0pt}}c}

\author[1, \dagger]{Lin Ha}
\author[1]{Chen Zhou}
\author[1, \dagger]{Haozheng Luo}
\author[1]{Han Liu}
\author[1]{Xiaoyu Zhang}
\author[1, \dagger]{Yan Chen}

\affiliation[1]{Northwestern University}

\contribution[\dagger]{Corresponding authors}

\abstract{
We propose \syb, a structure-based framework for TCR--pMHC binding prediction that eliminates the need for a predicted docked complex, and with it the placement error that dominates such predictions. Traditional structure-based predictors utilize the co-folded complex structure of the receptor and the antigen and interpret binding from its contact surface. Consequently, every downstream feature inherits any deviation the structure-prediction tool introduced. Methodologically, \sys folds the TCR and pMHC independently, encodes them as SE(3)-invariant residue graphs, and compares their learned representations without constructing a shared coordinate frame. A relative pose is therefore not an input of the model, so placement error cannot propagate into the prediction. Computationally, \sys preprocesses $\mathcal{O}(N{+}M)$ monomer folds where a complex-based predictor needs $\mathcal{O}(NM)$ operations. Experimentally, under epitope-disjoint evaluation with six retrained baselines, \sys is first on every metric under all four folding tools, by \textbf{0.9 to 14.4\%} in micro-AUROC, \textbf{18.5 to 25.3\%} in micro-AUPRC and \textbf{10.6 to 24.1\%} in MCC$^{\ast}$. Its ranking is also up to \textbf{85\%} more stable across folding tools. 

}

\date{\today}
\correspondence{Lin Ha at \email{linha2026@u.northwestern.edu}, Haozheng Luo at \email{hluo@u.northwestern.edu}, Yan Chen at \email{ychen@northwestern.edu}}

\begin{document}
\maketitle


\vspace{0.1in}
\section{Introduction}
\vspace{0.1in}
\label{sec:intro}
We propose \syb, a structure-based predictor of TCR--pMHC binding that \textbf{never forms a docked complex}. It folds each molecule once as an unbound monomer, encodes the two independently, and scores the pair by comparing SE(3)-invariant descriptors, so the relative pose of the two molecules never enters the computation. Because TCR--pMHC binding determines adaptive immune specificity \citep{garcia2005sees} while the candidate space is far too large for exhaustive experimental screening, \sys provides a scalable way to prioritize likely binders, including for unseen epitopes.

Existing structure-based TCR--pMHC predictors make this decision through a predicted docked complex, first folding or docking the two molecules into a common assembly and then inferring binding from the predicted interface, using residue contacts, geometric features, or structure-derived graphs \citep{slone2025stag, slone2025stagllm, li2025sagetcr, abeer2026structure, polster2026t2pmhc, siriarchawatana2026tcrlens}. The design follows the biology, since recognition happens at the interface and sequence alone does not explain specificity \citep{garcia2005sees, lagruta2018drivers}. However, it also makes the downstream predictor dependent on the accuracy of the upstream complex prediction. In particular, these methods typically retain only residues within a fixed distance of the opposing molecule \citep{slone2025stag, li2025sagetcr, abeer2026structure}. The predicted relative pose therefore determines not only the geometric features presented to the model, but also which residues and intermolecular contacts exist in its input graph. An error in rigid-body placement can thus change the input itself, rather than merely perturb its coordinates. Although \citet{polster2026t2pmhc} attribute their model's limitations to structural prediction error, whether predicted TCR--pMHC complexes are reliable in these task-critical quantities has not been systematically tested. 

\begin{figure}[htp]
\centering
\vspace{-0.2in}
\includegraphics[width=\linewidth]{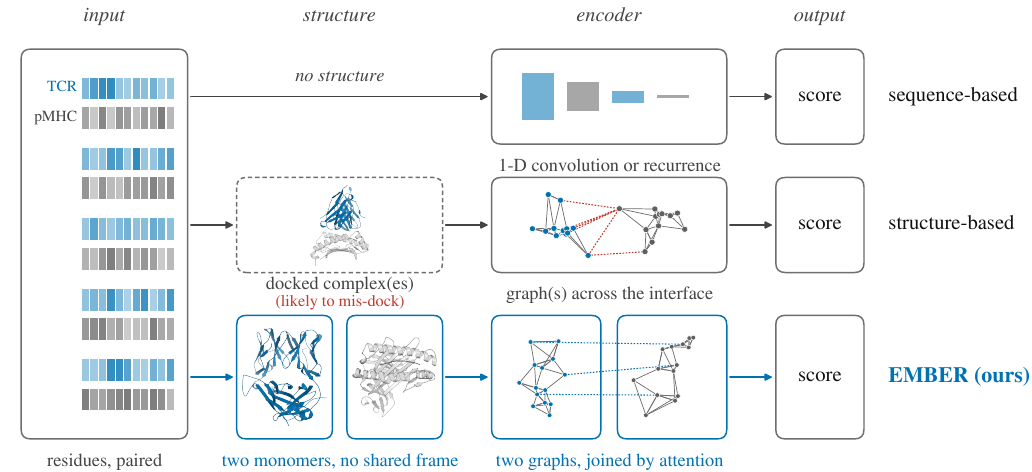}
\caption{\textbf{Three routes to binding prediction.} All methods score the same molecular pair. Sequence-based methods use no structural information; structure-based methods dock the molecules and construct a joint interface graph with pose-dependent contacts; \sys instead folds each molecule independently and connects the two monomer graphs only through attention.}
\label{fig:pipelines}
\vspace{-0.1in}
\end{figure}

To address this limitation, we propose \syb, which removes complex assembly rather than correcting its errors. Each molecule is folded as an unbound monomer and represented as a residue graph, with C$\alpha$-centered nodes carrying local geometric descriptors and cached protein language-model embeddings. Edges connect nearest neighbors and encode spatial separation and biochemical compatibility. Because all input features are invariant to rigid-body transformations, no absolute coordinates or inter-molecular pose enter the representation. Given these pose-independent graphs, a shared equivariant encoder processes the two molecules separately, restricting message passing to within each molecule and updating neighborhoods from the evolving coordinates at each layer. Cross-molecule interaction occurs only through bidirectional attention between the CDR3 loops and peptide, whose residue sets are defined by sequence annotation rather than predicted geometry. Thus, \sys neither constructs a shared coordinate frame nor selects residues using inter-molecular distance cutoffs. Pose independence is therefore imposed by the model architecture and is preserved throughout training.

\textbf{Contributions.} We present \syb (\cref{fig:pipelines}), a monomer-only TCR--pMHC structure-based predictor which is provably independent of the relative pose of the two molecules. Our main contributions are:

\begin{itemize}[leftmargin=*]
\item We identify \textbf{assembly-induced interface corruption (AIC)}, a previously uncharacterized failure mode in structure-based TCR--pMHC prediction, where pose errors corrupt the residues and contacts exposed to downstream predictors.
    \item We propose \sys, a pose-free structure-based predictor that folds the TCR and pMHC independently as unbound monomers. Each molecule is encoded as an SE(3)-invariant residue graph with a shared encoder, and binding is scored by bidirectional attention over the resulting representations. 
    \item \sys preprocesses $\mathcal{O}(N{+}M)$ monomer folds, against the $\mathcal{O}(NM)$ complex predictions a complex-based predictor needs.
    \item Experimentally, \sys leads six published baselines retrained on identical data and splits, first on every metric under all four folding tools. It is also significantly more stable to the choice of folding tool.

\end{itemize}


\vspace{0.05in}
\section{Related Work}
\vspace{0.1in}
\label{sec:related}
Existing approaches to TCR--pMHC binding prediction fall broadly into \textbf{sequence-based} and \textbf{structure-informed} families. Sequence models, including NetTCR \citep{montemurro2021nettcr, jensen2024nettcr22} and ERGO-II \citep{springer2021ergo2}, learn directly from amino-acid sequences and benefit from substantially larger sequence datasets than are available for solved TCR--pMHC complexes \citep{vita2015iedb,shugay2018vdjdb, leem2018stcrdab}. However, recent benchmarking and systematic evaluation show that generalization to unseen epitopes remains limited \citep{lu2026assessment, deng2026benchmarking}. TEIM predicts the residue contacts themselves from sequence and HeteroTCR builds a heterogeneous graph over sequence embeddings, so neither reads a structure \citep{peng2023teim, yu2024heterotcr}.
Advances in protein structure prediction enabled a second family of methods that incorporates predicted three-dimensional structure \citep{jumper2021alphafold, abramson2024alphafold3}. Most construct a predicted TCR--pMHC complex and infer binding from its interface using residue contacts, geometric features, and graph representations \citep{bradley2023tcrdock,slone2025stag, slone2025stagllm, li2025sagetcr, polster2026t2pmhc, siriarchawatana2026tcrlens}. These methods typically define the interface by a distance cutoff of 8 to 14~\AA{}, so the predicted pose determines not only residue geometry but also which residues and intermolecular contacts enter the model. Other approaches reduce this dependence by using the structure of only one binding partner \citep{zhao2023deepair,ji2024catcr, ge2025trap,shang2026imprint} or by using structure only during training \citep{li2026tcrsrim}, but do not exploit the structural information of both molecules. Prior work has already reported unreliable docking confidence \citep{deleuran2025nettcrstruc} and that non-interacting complexes adopt conformations geometrically similar to interacting ones \citep{robben2026aipredicted}.

\sys occupies a distinct point in this design space. STAG, SageTCR and the classifier of \citet{abeer2026structure} identify only the residues within 10 to 14~\AA{} of the opposing molecule as the interface \citep{slone2025stag,li2025sagetcr}. STAG-LLM keeps every residue and wires the assembly with a 10~\AA{} radius graph \citep{slone2025stagllm}, so there the pose decides the contacts rather than the residues. \sys processes each monomer in full and links the two only through sequence-annotated residue sets, so neither choice is exposed to the pose. TRAP avoids the assembly by dropping the receptor structure and representing the TCR by sequence \citep{ge2025trap}, while \sys keeps the structure of both. caTCR \citep{ji2024catcr} also folds its two inputs separately, though it folds only the CDR3$\beta$ loop and the epitope and leaves the MHC out of the model. It reduces each to a matrix of pairwise C$\alpha$ distances and reads that matrix with a convolutional network, whose neighborhoods follow the residue order rather than the spatial arrangement. \sys carries coordinates as an equivariant field and updates them at every layer. Running an E($n$)-equivariant network over the joint graph, as \citet{abeer2026structure} does, leaves the score invariant to motion of the whole assembly while the relative pose stays in the input. In contrast, \sys removes the relative pose by \cref{prop:pose}.

\vspace{0.05in}

\section{Assembly-induced Interface Corruption}
\vspace{0.1in}
\label{sec:analysis}
\begin{figure}[ht]
\centering
\includegraphics[width=\linewidth]{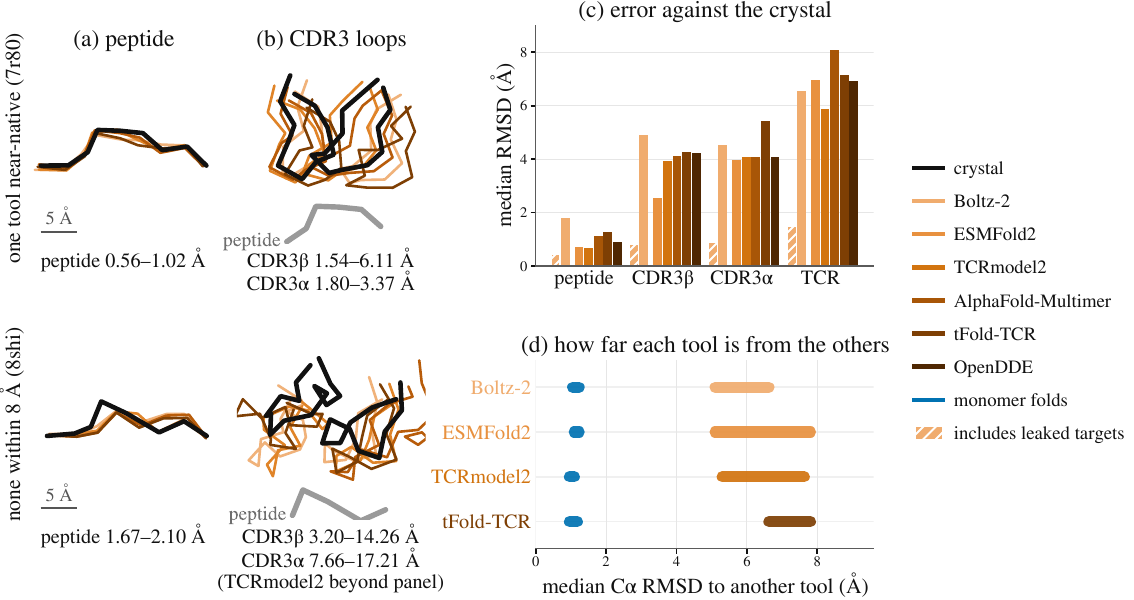}
\vspace{-0.1in}
\caption{\textbf{The predicted interface, and where its error lives.} \textbf{(a, b)} Two targets in the plane of the peptide, with the crystal drawn in black. The upper one lies inside Boltz-2's training window, the lower one is an example of generalization failure where no tool places the receptor within 8~\AA{}. \textbf{(c)} Median C$\alpha$ RMSD after superposing on the MHC heavy chain. Boltz-2 appears twice, hatched over all 49 targets and solid over the 10 released after its training cutoff. \textbf{(d)} Each bar spans one tool's pairings with the others, each tool's own color for receptor placement on the 49 complexes and blue for whole-monomer agreement over the same 250 TCR monomers. AlphaFold-Multimer and OpenDDE have no monomer folds. TCRmodel2 is the TCR-specialized version of AlphaFold-Multimer, and OpenDDE shares its AlphaFold3 architecture with Boltz-2, so each is already represented.}
\vspace{-0.2in}
\label{fig:interface}
\end{figure}

In this section, we test whether predicted TCR--pMHC complexes provide reliable structural inputs for downstream binding prediction. We evaluate six co-folding tools on 49 class-I TCR--pMHC crystal structures from STCRDab \citep{leem2018stcrdab}. The tools are Boltz-2 \citep{passaro2025boltz2}, ESMFold2 \citep{candido2026esm}, TCRmodel2 \citep{yin2023tcrmodel2}, AlphaFold-Multimer \citep{evans2021multimer}, OpenDDE \citep{aureka2026opendde} and tFold-TCR \citep{wu2025tfold}. We focus on two quantities consumed by structure-based predictors: the intermolecular contacts at the interface and the relative placement of the two molecules.

We term \textbf{assembly-induced interface corruption (AIC)} the failure mode in which rigid-body placement error changes the interface presented to the downstream model. Because traditional structure-based models construct interface graphs using distance cutoffs of 8--14~\AA{}, an error of comparable magnitude can alter which residues and intermolecular contacts enter the graph, producing misleading information. In short, we found that no tool reconstructs the native interface reliably. Setting Boltz-2 aside, the median contact Jaccard at 5~\AA{} runs from 0.188 to 0.459 and the median receptor C$\alpha$ RMSD from 5.87 to 8.09~\AA{}. Boltz-2 is best in almost every row and reaches a Jaccard of 0.833, but 39 of the 49 targets were seen by it during training. On the 10 targets released after Boltz-2's training cutoff, its Jaccard falls from 0.833 to 0.155. The other five tools move far less on those same targets, ESMFold2 from 0.459 to 0.393 and AlphaFold-Multimer from 0.308 to 0.303. Therefore the lead of Boltz-2 reflects memorization rather than generalization ability. \Cref{tab:folding} and \cref{tab:folding-clean} list the details.

More specifically, we locate the error in assembly rather than in monomer folding. Peptides remain close to their crystal positions, whereas CDR3 loops and the receptor as a whole are displaced substantially further (\cref{fig:interface}c). Across tools, receptor placements differ by 5.10 to 7.80~\AA{} on the same complexes, while independently folded TCR monomers agree to within 0.96 to 1.23~\AA{} over the same 250 TCR monomers (\cref{fig:interface}d). Even the two tools that agree most on placement sit 5.10~\AA{} apart, four times the widest monomer disagreement. The tools largely agree on the individual molecular structures and disagree on how to assemble them.

Most co-folding tools are stochastic, so one natural solution is to fold each complex $K$ times and average the error away. Write the AIC of sample $k$ as a displacement vector $v_k$ and define \emph{directional coherence} $\gamma=\big\lVert\tfrac{1}{K}\sum_{k=1}^{K} v_k/\lVert v_k\rVert\big\rVert\in[0,1]$, which is $1$ when every sample misses the same way. When the displacements have similar length, averaging $K$ samples leaves at least a fraction $\gamma$ of the error, so measuring $\gamma$ shows what sampling can achieve. Five tools fold each target $20$ times and reach $\gamma$ of $0.84$ to $0.992$, four times the $0.21$ expected from random directions (tFold-TCR is deterministic and offers no distribution to sample). Equal-weight averaging accordingly moves the displacement by $-2.5\%$ to $+6.7\%$. Moreover, this bias is not specific to one tool. In a common frame defined by the peptide-binding groove, five independently built tools share one displacement direction, with a cross-tool $\gamma$ of $0.621$ and agreement beyond random on $40$ of $48$ targets ($p=3.3\times10^{-6}$). Therefore, we claim that assembly-induced interface corruption is a systematic bias.

\finding{The dominant structural error in predicted TCR--pMHC complexes lies in rigid-body assembly. It is large enough to change which residues and contacts enter a downstream graph, and its direction is systematic, so multiple sampling and ensembling cannot remove it.}

\vspace{0.05in}

\section{Method}
\vspace{0.1in}
\label{sec:method}
\begin{figure}[htp]
\centering
\vspace{-0.1in}
\includegraphics[width=\linewidth]{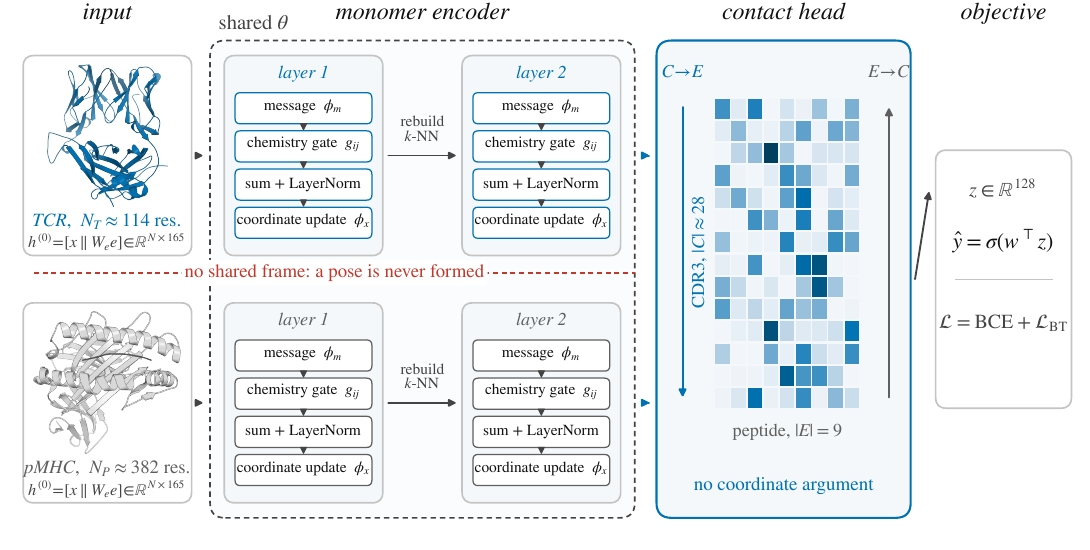}
\vspace{-0.3in}
\caption{\textbf{\sys.} The top row is the TCR and the bottom row the pMHC. Tensor shapes are annotated on the flow, and both encoder layers are drawn. In the contact head, $C$ indexes the CDR3$\beta$ residues and $E$ the peptide residues, so the map carries one row per CDR3$\beta$ residue and one column per peptide residue. Its values are schematic. The dashed box encloses both towers because they share weights and exchange no message.}
\label{fig:method}
\vspace{-0.05in}
\end{figure}

We propose \syb (\cref{fig:method}), a framework for TCR--pMHC binding prediction that folds each molecule independently and represents it as a residue graph with SE(3)-invariant features. A shared-weight encoder processes the two graphs separately. A bidirectional attention head then compares and pools their residue representations, and a linear layer produces the binding score. \Cref{alg:ember} states the forward pass as pseudocode, and \cref{tab:notation} introduces the symbols in detail.

\textbf{Problem Definition.}
\sys represents the receptor and antigen as residue graphs, $\mathcal{G}^{T}$ and $\mathcal{G}^{P}$. The receptor graph carries both TCR chains, and the encoder is given the $\beta$ chain alone, $N_T \approx 114$ nodes (\cref{app:scope}). The antigen graph holds the MHC heavy chain, $\beta_2$-microglobulin and the peptide, $N_P \approx 382$ nodes. Each residue $i$ is described by its C$\alpha$ coordinate $p_i \in \mathbb{R}^3$, an invariant local descriptor $x_i \in \mathbb{R}^{37}$, and a frozen ESM-C embedding $e_i \in \mathbb{R}^{1152}$ \citep{candido2026esm}. The descriptor $x_i$ depends only on local geometry and residue identity. It is therefore invariant to rigid motion, and it introduces no absolute coordinate and no preferred axis. The encoder input is $h^{(0)}_i = [x_i \,\|\, W_e e_i] \in \mathbb{R}^{165}$ and $p^{(0)}_i = p_i$, with $W_e \in \mathbb{R}^{128\times1152}$ learned. Graph edges connect the $k=30$ nearest neighbors in C$\alpha$ distance $d_{ij}=\lVert p_i-p_j\rVert$. Each edge carries $a_{ij} = [\rho(d_{ij}) \,\|\, d_0/d_{ij} \,\|\, b_{ij}] \in \mathbb{R}^{22}$, where $\rho$ is a 16-dimensional radial basis over $[0,20]$~\AA{} and $d_0 = 4$~\AA{}. The five entries of $b_{ij}$ read residue identity alone: hydrogen bonding, hydrophobicity, van der Waals radius, charge product, and aromaticity.

\textbf{Monomer Encoder.}
A single encoder $\theta$ runs on the two graphs without any message between them. It first projects $h^{(0)}_i$ to the working width $D=256$. Every layer then repeats the same four steps in order, a message on each edge, a chemistry gate on that message, a node update and a coordinate update, and closes by rebuilding the neighborhood from the new coordinates. Writing $\mathcal{N}(i)$ for the neighbors of $i$, $\sigma$ for the logistic sigmoid, $\mathrm{LN}$ for layer normalization and $c_{ij}=[\,h^{(\ell-1)}_i \,\|\, h^{(\ell-1)}_j \,\|\, a_{ij}\,]$ for the edge context, layer $\ell$ maps $(h^{(\ell-1)},p^{(\ell-1)})\mapsto(h^{(\ell)},p^{(\ell)})$ by
\begin{align}
m_{ij} &= \phi_m(c_{ij}), &
g_{ij} &= \sigma\big(\phi_g(b_{ij})\big)\in(0,1), \label{eq:msg}\\[2pt]
h^{(\ell)}_i &= \mathrm{LN}\Big(h^{(\ell-1)}_i+\!\!\sum_{j\in\mathcal{N}(i)}\! g_{ij}\,m_{ij}\Big), &
p^{(\ell)}_i &= p^{(\ell-1)}_i+\frac{1}{|\mathcal{N}(i)|}\!\!\sum_{j\in\mathcal{N}(i)}\!\!\big(p^{(\ell-1)}_i-p^{(\ell-1)}_j\big)\,\phi_x(c_{ij}), \label{eq:upd}
\end{align}
after which the edge set is rebuilt from the updated coordinates as $\mathrm{kNN}(p^{(\ell)},k)$. This is the E($n$)-equivariant layer of \citet{satorras2021egnn} with two changes. First, the gate $g_{ij}$ reads residue identity only. Chemistry therefore modulates how much geometry is passed, and equivariance is untouched. Second, the graph is rebuilt at every layer. Coordinates reach $h$ only through $d_{ij}$ inside $a_{ij}$, so the updated coordinates would never influence $h$ without the rebuild.

\begin{algorithm}[t]
\caption{\sys forward pass.}
\label{alg:ember}
\small
\begin{algorithmic}[1]
\REQUIRE monomer graphs $\mathcal{G}^{T},\mathcal{G}^{P}$; encoder $\theta$; depth $L$; neighbors $k$
\FOR{$\mathcal{G}\in\{\mathcal{G}^{T},\mathcal{G}^{P}\}$}
  \STATE $h^{(0)}\leftarrow[\,x\,\|\,W_e e\,]$; \; $p^{(0)}\leftarrow$ C$\alpha$ coordinates \COMMENT{shared $\theta$, no message between graphs}
  \FOR{$\ell=1$ \TO $L$}
    \STATE $\mathcal{K}^{(\ell-1)}\leftarrow\mathrm{kNN}(p^{(\ell-1)},k)$; \; $h^{(\ell)},p^{(\ell)}\leftarrow\mathrm{EGCL}_{\theta}\big(h^{(\ell-1)},p^{(\ell-1)},\mathcal{K}^{(\ell-1)}\big)$ \COMMENT{Eq.~\ref{eq:msg}--\ref{eq:upd}}
  \ENDFOR
  \STATE $H\leftarrow h^{(L)}$ stacked by residue
\ENDFOR
\STATE $C\leftarrow$ CDR3 rows of $H^{T}$; \; $E\leftarrow$ peptide rows of $H^{P}$ \COMMENT{invariant scalars only}
\STATE $z_t\leftarrow\mathrm{proj}_1(\mathrm{pool}_{C}(\mathrm{MHA}_1(H^{T}_{C},H^{P}_{E})))$; \; $z_p\leftarrow\mathrm{proj}_2(\mathrm{pool}_{E}(\mathrm{MHA}_2(H^{P}_{E},H^{T}_{C})))$
\STATE $z\leftarrow\mathrm{LN}(W_z[\,z_t\,\|\,z_p\,])$
\RETURN $\hat{y}=\sigma(u)$, \; $u=w^{\top}z$ \COMMENT{no coordinate enters the contact head}
\end{algorithmic}
\end{algorithm}

\textbf{Contact Head.}
The two molecules meet in the bidirectional contact head that takes no coordinate argument. Write $H^{T}\in\mathbb{R}^{N_T\times D}$ and $H^{P}\in\mathbb{R}^{N_P\times D}$ for the final states $h^{(L)}$ stacked by residue. Let $C$ index the CDR3$\beta$ rows of the receptor and $E$ the peptide rows of the antigen. Both index sets come from sequence annotation rather than from the predicted geometry, so they cannot be altered by AIC, unlike a distance cutoff. Both directions use the standard attention of \citet{vaswani2017attention} with four heads of width $32$. \textbf{TCR to peptide}, $\mathrm{MHA}_1(H^{T}_{C},H^{P}_{E})$, lets the CDR3$\beta$ rows query the peptide. \textbf{Peptide to TCR}, $\mathrm{MHA}_2(H^{P}_{E},H^{T}_{C})$, lets the peptide rows query the CDR3$\beta$ rows. Their per-query outputs are pooled over $C$ and over $E$ and projected, giving $z_t=\mathrm{proj}_1(\mathrm{pool}_{C}(\cdot))$ and $z_p=\mathrm{proj}_2(\mathrm{pool}_{E}(\cdot))$, both in $\mathbb{R}^{128}$. The pair is scored from their concatenation, $z=\mathrm{LN}(W_z[\,z_t\|z_p\,])\in\mathbb{R}^{128}$, logit $u=w^{\top}z$, and $\hat{y}=\sigma(u)$.

TCR to peptide pools over $C$ rather than over $E$, so peptide length enters only through the attention over keys and never sets the pooled dimension. It carries little receptor-side signal of its own, and it needs no canonical frame. In contrast, Peptide to TCR needs a canonical frame, because peptide length ranges from 8 to 11 residues in our dataset and pooling over raw positions would give position $k$ a different chemical role. It therefore pools on a canonical frame of $S=9$ slots. Class-I binding pins the peptide at both ends regardless of length, the free N-terminus in the A pocket, the second residue P2 in the B pocket and the C-terminal residue P$\Omega$ in the F pocket \citep{madden1995pmhc,szeto2021rulemakers}. A monotone map $\pi:\{1,\dots,|E|\}\to\{1,\dots,S\}$ sends those three to slot~1, slot~2 and slot~$S$, and spreads the interior residues over the remaining slots, placing surplus residues and gaps centrally, where a longer peptide bulges away from the groove \citep{tynan2005superbulged}. An 8-mer and an 11-mer therefore present the same anchor at the same slot, and length is absorbed where class-I presentation already tolerates it. Writing $O$ for the per-residue output of $\mathrm{MHA}_2$, $\mathrm{pool}_{E}(O)=\big[\,q_1\|\cdots\|q_S\,\big],
q_s=\big[\,\mathrm{mean}_{\pi(i)=s}\,O_i \;\big\|\; \max_{\pi(i)=s}\,O_i\,\big]$, and $\mathrm{proj}_2:\mathbb{R}^{2S\cdot 128}\to\mathbb{R}^{128}$.

\textbf{Objective Function.} Let a batch $\mathcal{B}=\mathcal{B}^{+}\cup\mathcal{B}^{-}$ hold binders $\mathcal{B}^{+}$ and non-binders $\mathcal{B}^{-}$, and let $y_i\in\{0,1\}$ be the label of pair $i$. \sys is trained on $\mathcal{L}=\mathcal{L}_{\mathrm{BCE}}+\mathcal{L}_{\mathrm{BT}}$ with
\begin{align*}
\mathcal{L}_{\mathrm{BCE}} &= -\frac{1}{|\mathcal{B}|}\sum_{i\in\mathcal{B}}\big[y_i\log\hat{y}_i+(1-y_i)\log(1-\hat{y}_i)\big], &\!\!\!\!\!
\mathcal{L}_{\mathrm{BT}} &= -\frac{1}{|\mathcal{B}^{+}||\mathcal{B}^{-}|}\!\!\!\sum_{a\in\mathcal{B}^{+}\!,\,b\in\mathcal{B}^{-}}\!\!\!\log\sigma(u_a-u_b).
\end{align*}
The first term calibrates the per-pair probability. The second ranks candidates within the batch \citep{bradley1952rank}. See \cref{app:training} for detailed hyperparameters.

\vspace{0.1in}
\subsection{Theoretical Analysis}
\label{sec:pose}
\vspace{0.1in}
A rigid motion $(R,t)$ acts on a residue graph by $p_i\mapsto Rp_i+t$, leaving residue identities untouched.

\begin{definition}[encoder inputs]
\label{def:inputs}
The encoder input $h^{(0)}$, the edge features $a$ and the $k$-NN edge set are unchanged by $(R,t)$.
\end{definition}
\begin{lemma}[one layer]
\label{lem:layer}
If $h^{(\ell-1)}$ is unchanged by $(R,t)$ and $p^{(\ell-1)}\mapsto Rp^{(\ell-1)}+t$, then layer $\ell$ leaves $h^{(\ell)}$ unchanged and sends $p^{(\ell)}\mapsto Rp^{(\ell)}+t$.
\end{lemma}
\vspace{-0.1in}
\begin{proof}
See \cref{app:lemma2}.
\end{proof}
\vspace{-0.1in}
\begin{proposition}[Pose independence]
\label{prop:pose}
Let $s(\mathcal{G}^{T},\mathcal{G}^{P})=\hat{y}$ be the binding score. For any rotations $R_1,R_2\in SO(3)$ and translations $t_1,t_2\in\mathbb{R}^{3}$ applied \emph{independently} to the two molecules,
\begin{align*}
s\big(R_1\mathcal{G}^{T}\!+t_1,\;R_2\mathcal{G}^{P}\!+t_2\big)=s\big(\mathcal{G}^{T},\mathcal{G}^{P}\big).
\end{align*}
\end{proposition}
\vspace{-0.1in}
\begin{proof}
See \cref{app:prop}.
\end{proof}
\vspace{-0.1in}

\sys therefore never takes in the relative pose, so AIC cannot mislead the result.

\vspace{0.05in}

\section{Experiments}
\vspace{0.1in}
\label{sec:exp}
We evaluate \sys against six published baselines on a curated set of four public datasets, over five epitope-disjoint folds and three seeds. 

\textbf{Dataset.} We utilize four public datasets---McPAS-TCR, VDJdb, IEDB, and 10x Genomics---covering 45,909 TCR--pMHC pairs across 343 epitopes~\citep{vita2015iedb,tickotsky2017mcpas,shugay2018vdjdb,tenx2019immunity}. Following \citet{slone2025stag}, 7,371 measured binders are augmented with sampled non-binders, yielding a binder prevalence of 0.161.

\textbf{Baselines.} To evaluate \sys against diverse TCR--pMHC predictors, we compare with six baselines spanning three structural, one hybrid, and two sequence-based methods. Structural baselines include STAG-LLM~\citep{slone2025stagllm}, SageTCR~\citep{li2025sagetcr}, and the dual-encoder classifier of \citet{abeer2026structure}; TRAP~\citep{ge2025trap} combines pMHC structure with TCR sequence, whereas NetTCR-2.2~\citep{jensen2024nettcr22} and ERGO-II~\citep{springer2021ergo2} read no structure. We retrain each method using its published hyperparameters across five folds and three random seeds, reporting means over all 15 runs and standard deviations across seeds.

\textbf{Metrics.} To evaluate the performance of \sys, we measure ranking quality, classification accuracy, and cross-tool stability. Ranking is assessed by micro- and macro-averaged AUROC and AUPRC; classification uses MCC at a fixed threshold of 0.5, reported alongside MCC at the threshold that maximizes it on the test fold, which no deployed system has and which we report only to bound what threshold tuning could buy a baseline. Stability is quantified by pairwise Kendall's $\tau_b$, Kendall's concordance $W$ across all four folding tools, and Overlap@$k$ between top-ranked candidates. \cref{app:metrics} provides the definitions.

\vspace{0.2in}
\begin{table}[htp]
\centering
\vspace{-0.25in}
\footnotesize\setlength{\tabcolsep}{2.6pt}
\caption{\textbf{Binding prediction under epitope-disjoint five-fold cross-validation}, means over five folds and three seeds, with the standard deviation over seeds. \emph{structures} names the tool that produced the input, and within each group the best value per column is \textbf{bold} and the second best is \underline{underlined}. The sequence baselines read no structure and are not ranked. \sys is first in every cell, ahead of the second by \textbf{6.1 to 21.7\%} in macro-AUROC, \textbf{13.6 to 22.0\%} in macro-AUPRC, \textbf{0.9 to 14.4\%} in micro-AUROC, \textbf{18.5 to 25.3\%} in micro-AUPRC, \textbf{51.2 to 74.5\%} in MCC@0.5 and \textbf{10.6 to 24.1\%} in MCC$^{\ast}$. MCC$^{\ast}$ is MCC at the threshold that maximizes it on the test fold. \Cref{app:perfold} gives the per-fold values.}
\label{tab:main}
\vspace{-0.05in}
\resizebox{\textwidth}{!}{
\begin{tabular}{lcccccccc}
\toprule
 family & structures & method & macro-AUROC & macro-AUPRC & micro-AUROC & micro-AUPRC & MCC@0.5 & MCC$^{\ast}$ \\
\midrule
 \multirow{20}{*}{structure-based} & \multirow{5}{*}{tFold-TCR} & \syb & \textbf{0.758}\,\sd{0.010} & \textbf{0.585}\,\sd{0.008} & \textbf{0.666}\,\sd{0.003} & \textbf{0.366}\,\sd{0.014} & \textbf{0.251}\,\sd{0.059} & \textbf{0.292}\,\sd{0.017} \\
& & STAG-LLM & 0.620\,\sd{0.018} & 0.473\,\sd{0.020} & \underline{0.582}\,\sd{0.008} & \underline{0.292}\,\sd{0.006} & \underline{0.166}\,\sd{0.024} & \underline{0.264}\,\sd{0.018} \\
& & SageTCR & 0.622\,\sd{0.009} & 0.443\,\sd{0.004} & 0.547\,\sd{0.011} & 0.240\,\sd{0.007} & 0.122\,\sd{0.027} & 0.182\,\sd{0.036} \\
& & Abeer & 0.591\,\sd{0.008} & 0.434\,\sd{0.006} & 0.550\,\sd{0.011} & 0.225\,\sd{0.014} & 0.071\,\sd{0.007} & 0.157\,\sd{0.024} \\
& & TRAP & \underline{0.651}\,\sd{0.015} & \underline{0.513}\,\sd{0.014} & 0.499\,\sd{0.010} & 0.182\,\sd{0.010} & 0.034\,\sd{0.000} & 0.095\,\sd{0.024} \\
\addlinespace[3pt]\cmidrule(lr){2-9}\addlinespace[3pt]
& \multirow{5}{*}{ESMFold2} & \syb & \textbf{0.759}\,\sd{0.004} & \textbf{0.584}\,\sd{0.003} & \textbf{0.670}\,\sd{0.005} & \textbf{0.372}\,\sd{0.013} & \textbf{0.273}\,\sd{0.029} & \textbf{0.304}\,\sd{0.011} \\
& & STAG-LLM & \underline{0.714}\,\sd{0.013} & 0.507\,\sd{0.016} & \underline{0.659}\,\sd{0.003} & \underline{0.312}\,\sd{0.007} & \underline{0.172}\,\sd{0.001} & \underline{0.245}\,\sd{0.004} \\
& & SageTCR & 0.618\,\sd{0.002} & 0.437\,\sd{0.014} & 0.545\,\sd{0.015} & 0.227\,\sd{0.003} & 0.115\,\sd{0.016} & 0.174\,\sd{0.030} \\
& & Abeer & 0.590\,\sd{0.003} & 0.431\,\sd{0.002} & 0.579\,\sd{0.008} & 0.257\,\sd{0.028} & 0.108\,\sd{0.049} & 0.188\,\sd{0.025} \\
& & TRAP & 0.649\,\sd{0.006} & \underline{0.514}\,\sd{0.012} & 0.497\,\sd{0.025} & 0.180\,\sd{0.003} & 0.041\,\sd{0.017} & 0.096\,\sd{0.008} \\
\addlinespace[3pt]\cmidrule(lr){2-9}\addlinespace[3pt]
& \multirow{5}{*}{TCRmodel2} & \syb & \textbf{0.769}\,\sd{0.007} & \textbf{0.593}\,\sd{0.004} & \textbf{0.668}\,\sd{0.003} & \textbf{0.367}\,\sd{0.002} & \textbf{0.267}\,\sd{0.009} & \textbf{0.301}\,\sd{0.003} \\
& & STAG-LLM & \underline{0.632}\,\sd{0.009} & 0.477\,\sd{0.005} & 0.586\,\sd{0.008} & \underline{0.295}\,\sd{0.018} & \underline{0.153}\,\sd{0.029} & \underline{0.264}\,\sd{0.019} \\
& & SageTCR & 0.596\,\sd{0.028} & 0.418\,\sd{0.011} & 0.537\,\sd{0.012} & 0.235\,\sd{0.004} & 0.112\,\sd{0.018} & 0.174\,\sd{0.011} \\
& & Abeer & 0.583\,\sd{0.023} & 0.434\,\sd{0.017} & \underline{0.597}\,\sd{0.006} & 0.291\,\sd{0.021} & 0.145\,\sd{0.024} & 0.236\,\sd{0.019} \\
& & TRAP & 0.630\,\sd{0.008} & \underline{0.486}\,\sd{0.011} & 0.487\,\sd{0.009} & 0.172\,\sd{0.004} & 0.019\,\sd{0.006} & 0.077\,\sd{0.006} \\
\addlinespace[3pt]\cmidrule(lr){2-9}\addlinespace[3pt]
& \multirow{5}{*}{Boltz-2} & \syb & \textbf{0.753}\,\sd{0.004} & \textbf{0.579}\,\sd{0.007} & \textbf{0.661}\,\sd{0.006} & \textbf{0.372}\,\sd{0.003} & \textbf{0.278}\,\sd{0.013} & \textbf{0.311}\,\sd{0.012} \\
& & STAG-LLM & \underline{0.710}\,\sd{0.002} & \underline{0.505}\,\sd{0.003} & \underline{0.655}\,\sd{0.008} & \underline{0.314}\,\sd{0.026} & \underline{0.169}\,\sd{0.012} & \underline{0.254}\,\sd{0.020} \\
& & SageTCR & 0.597\,\sd{0.001} & 0.422\,\sd{0.013} & 0.527\,\sd{0.012} & 0.224\,\sd{0.009} & 0.115\,\sd{0.003} & 0.146\,\sd{0.004} \\
& & Abeer & 0.600\,\sd{0.003} & 0.437\,\sd{0.009} & 0.580\,\sd{0.014} & 0.255\,\sd{0.032} & 0.054\,\sd{0.026} & 0.188\,\sd{0.023} \\
& & TRAP & 0.642\,\sd{0.015} & 0.502\,\sd{0.006} & 0.491\,\sd{0.011} & 0.175\,\sd{0.002} & 0.016\,\sd{0.013} & 0.082\,\sd{0.009} \\
\midrule
\multirow{2}{*}{sequence-based} & - & NetTCR-2.2 & 0.688\,\sd{0.007} & 0.536\,\sd{0.012} & 0.604\,\sd{0.010} & 0.319\,\sd{0.009} & 0.207\,\sd{0.026} & 0.257\,\sd{0.006} \\
& - & ERGO-II & 0.668\,\sd{0.016} & 0.502\,\sd{0.016} & 0.593\,\sd{0.009} & 0.284\,\sd{0.010} & 0.061\,\sd{0.026} & 0.224\,\sd{0.012} \\
\bottomrule
\end{tabular}
}
\end{table}

\subsection{Binding Prediction and Stability against Retrained Baselines}
\label{sec:exp:main}
\vspace{0.1in}
We conduct experiments on five epitope-disjoint folds, ensuring that all test peptides are unseen during training. We run the comparison on structures generated by four different tools, retraining every method on that tool's structures (\cref{tab:main,fig:results}). A structure-based predictor inherits its upstream tool, and the folding tools disagree, so we also measure how far that choice moves the answer. For each epitope and fold we rank that epitope's candidate TCRs with the same method trained on each tool's structures, then compare the rankings (\cref{tab:stability}).\\ 
\textbf{Results.} \sys is first on every metric under every folding tool, and the size of the lead depends on which part of the ranking a metric reads. AUROC weighs the whole candidate list, and there the strongest baseline comes close. STAG-LLM trails \sys by 0.011 in micro-AUROC under ESMFold2 and by 0.006 under Boltz-2, so the whole-list metric barely separates the two. Metrics that weigh the top of the list give a different reading. On micro-AUPRC the same baseline trails by 0.058 to 0.074, and on MCC at a fixed threshold of 0.5 it trails by 51 to 75\%. Calibration accounts for part of that last figure, since tuning the threshold on the test fold narrows the lead to 10 to 24\%. What remains after the tuning is discrimination, and a deployed screen meets both terms at once, because it reads the top of the list at a threshold fixed in advance. Read one fold at a time, the lead holds on all five folds in 16 of the 24 cells, and the eight that fall short are two micro-AUROC cells, two micro-AUPRC cells and all four $\mathrm{MCC}^{\ast}$ cells (\cref{tab:paired}).

\begin{figure}[htp]
\centering
\includegraphics[width=\linewidth]{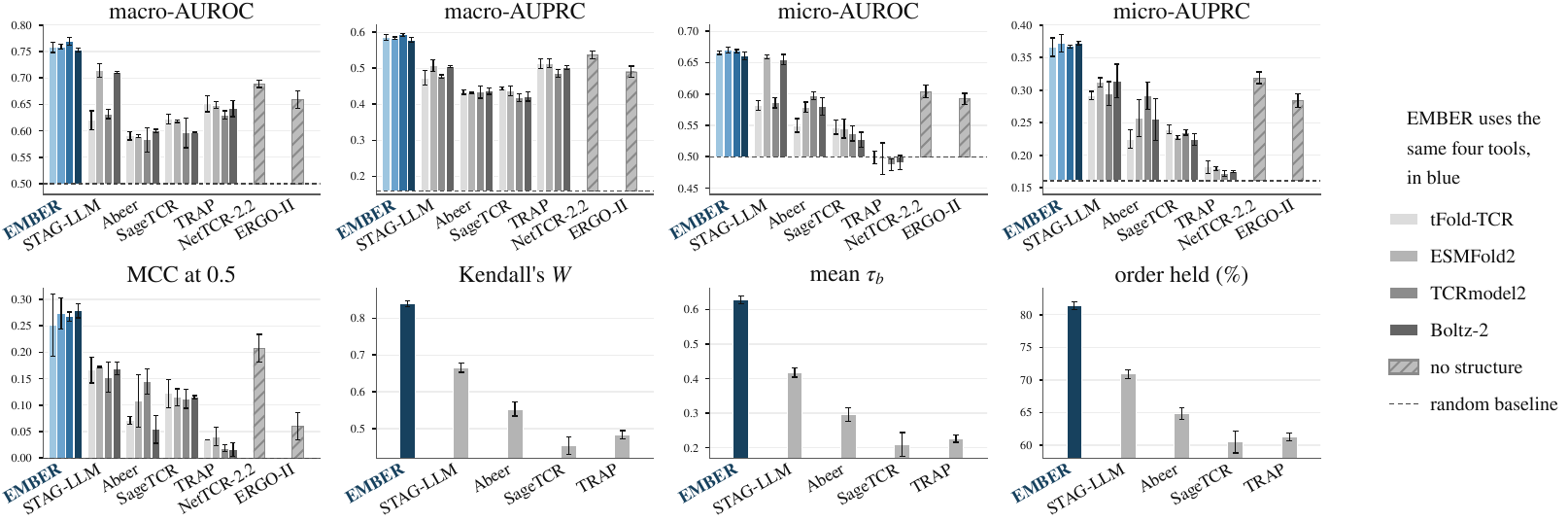}
\caption{\textbf{Means over five folds and three seeds, with every bar starting at chance}. Structural methods carry one bar per folding tool, light to dark in the order tFold-TCR, ESMFold2, TCRmodel2, Boltz-2. A method whose four bars are uneven depends on the upstream tool. The last three panels read the same rankings three ways, over the four tools of \cref{tab:stability}. Kendall's $W$ carries bootstrapped intervals.}
\label{fig:results}
\vspace{-0.1in}
\end{figure}

\begin{table}[htp]
\centering
\footnotesize
\setlength{\tabcolsep}{6pt}

\begin{minipage}[t]{0.28\textwidth}
    \vspace{0pt}
    \captionsetup{justification=raggedright,singlelinecheck=false}
    \caption{\textbf{Cross-tool agreement of methods trained on structures from four folding tools.}
    Agreement is evaluated over 142 epitopes with at least ten candidates each, comprising 44.6k candidates in total.}
    \label{tab:stability}
\end{minipage}
\hfill
\begin{minipage}[t]{0.71\textwidth}
    \vspace{0pt}
    \centering
    \resizebox{\linewidth}{!}{%
    \begin{tabular}{lcccc}
        \toprule
        method & $W$ & mean $\tau_b$ & order held & Overlap@10 \\
        \midrule
        \syb
        & \textbf{0.840}\,\sd{0.008}
        & \textbf{0.628}\,\sd{0.012}
        & \textbf{81.4\%}\,\sd{0.6}
        & \textbf{0.729}\,\sd{0.013} \\
        \midrule
        STAG-LLM
        & 0.666\,\sd{0.013}
        & 0.418\,\sd{0.014}
        & 70.9\%\,\sd{0.7}
        & 0.675\,\sd{0.008} \\
        Abeer
        & 0.553\,\sd{0.019}
        & 0.296\,\sd{0.019}
        & 64.8\%\,\sd{0.9}
        & 0.599\,\sd{0.001} \\
        TRAP
        & 0.484\,\sd{0.011}
        & 0.227\,\sd{0.011}
        & 61.3\%\,\sd{0.6}
        & 0.515\,\sd{0.010} \\
        SageTCR
        & 0.454\,\sd{0.023}
        & 0.210\,\sd{0.034}
        & 60.5\%\,\sd{1.7}
        & 0.573\,\sd{0.046} \\
        \bottomrule
    \end{tabular}%
    }
\end{minipage}
\end{table}

\vspace{-0.1in}
\sys is the most stable method by a wide margin. Its concordance is $W=0.840$, against 0.454 to 0.666 for the four structure-consuming baselines, and that is \textbf{26\% to 85\% higher agreement} than any of them. Read pairwise, 81.4\% of candidate pairs keep their relative order when the folding tool is swapped, against 60.5\% to 70.9\% for the baselines. The spread across seeds says what the comparison can resolve. \sys's $W$ varies by 0.008 across seeds and the closest baseline by 0.013, against a separation of 0.174. The least stable baseline varies most, up to 0.023 in $W$ and 0.034 in $\tau_b$, so an experiment expecting an effect that size needs more seeds rather than more folds.

\finding{\sys leads every metric under every folding tool, most clearly where every epitope counts equally and at the threshold a deployed system runs at. Its ranking is also the least disturbed by a change of folding tool, by a margin far outside the spread across training seeds.}
\vspace{0.1in}
\subsection{Influence of Structure}
\label{sec:exp:structure}
\vspace{0.1in}
We conduct three analyses to assess how structural information influences prediction. First, we directly use each folding tool's interface confidence as a training-free binding score. Second, we stratify co-folded pairs into confidence percentile bands and evaluate each method at its own operating threshold (\cref{fig:structure}); because confidence is available only for co-folded complexes, only the three complex-based methods see the structure a band scores. For each band, we report the normalized positive likelihood ratio, $\mathrm{LR}^{+}=\mathrm{TPR}/\mathrm{FPR}$, relative to the method's lowest-confidence band, thereby removing differences in operating point. Finally, we ablate \sys to quantify the contribution of structural information.

\textbf{Results.} The complex carries evidence that two monomers do not. Within one pMHC, the interface confidence separates a true binder from a decoy at Cohen's $d$ of 0.20 to 0.60. The independently folded receptor reaches only 0.05 to 0.12 on the same pairs. Read untrained as a score, the same confidence beats all six retrained baselines under TCRmodel2. No method that reads the complex turns this evidence into discrimination. Every method gains from the least to the most confident band. The two sequence baselines see no structure, so their gain measures what the pair alone is worth. The three complex consumers do not exceed that level, averaging 16\% against 17\%. The metric itself is not saturated, since \sys reaches 56\% without reading the complex. Nor are the flat curves an artefact of weak classifiers, since STAG-LLM has the second highest micro-AUROC of the seven methods and the second lowest gain.

Removing the structure entirely from \sys costs 11.1\% of that margin in micro-AUPRC. Replacing the two monomers with the predicted complex costs 12.3\%, and it costs more still at a fixed threshold. A docked assembly is therefore worth no more to \sys than no structure at all. Removing the canonical peptide frame causes MCC@0.5 to fall by 23.4\%, and the fall is 8.3\% once the threshold is tuned on the test fold. So, the canonical frame improves calibration and per-epitope precision.

\begin{table}[H]
\vspace{-4pt}
\centering
\footnotesize\setlength{\tabcolsep}{4pt}
\caption{\textbf{Main ablation, on the ESMFold2 structures.} \emph{remove
structure} keeps the residue identities and the original ESM-C output. \emph{docked complex} feeds \sys the co-folded complex. \emph{no peptide frame} removes the canonical peptide frame. The arms run on ESMFold2 because it gives the most favourable input, leading every interface and placement metric on the post-cutoff targets of \cref{tab:folding-clean}. \Cref{app:ablations} gives two further ablations.}
 \vspace{-0.05in}
\resizebox{\textwidth}{!}{%
\begin{tabular}{lcccccc}
\toprule
arm & macro-AUROC & macro-AUPRC & micro-AUROC & micro-AUPRC & MCC@0.5 & MCC$^{\ast}$ \\
\midrule
\sys (control) & 0.759\,\sd{0.004} & 0.584\,\sd{0.003} & 0.670\,\sd{0.005} & 0.372\,\sd{0.013} & 0.273\,\sd{0.029} & 0.304\,\sd{0.011} \\
\quad remove structure & 0.753\,\sd{0.008} & 0.570\,\sd{0.006} & 0.649\,\sd{0.003} & 0.349\,\sd{0.001} & 0.236\,\sd{0.018} & 0.283\,\sd{0.010} \\
\quad docked complex & 0.745\,\sd{0.009} & 0.578\,\sd{0.013} & 0.654\,\sd{0.009} & 0.346\,\sd{0.033} & 0.232\,\sd{0.063} & 0.274\,\sd{0.027} \\
\quad no peptide frame & 0.740\,\sd{0.008} & 0.558\,\sd{0.004} & 0.664\,\sd{0.007} & 0.360\,\sd{0.004} & 0.209\,\sd{0.017} & 0.279\,\sd{0.009} \\
\bottomrule
\end{tabular}%
}

\label{tab:ablation}
\vspace{-12pt}
\end{table}

\finding{Co-folding encodes binding evidence in its own confidence, yet no complex consumer converts it into discrimination, and supplying the complex to \sys costs accuracy. The structure contributes a geometric scaffold, so two unbound monomers give up nothing measurable.}

\begin{figure}[!ht]
\centering
\begin{minipage}[t]{0.452\textwidth}
\vspace{0pt}
\centering
\includegraphics[width=\linewidth]{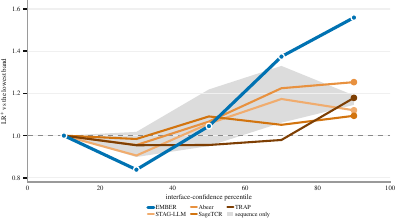}
\caption{\textbf{How $\mathrm{LR}^{+}$ varies with the structure confidence level.} The gray band spans the two sequence baselines. The first two bands hold the same share of binders, 0.112 each, so the step between them carries no signal.}
\label{fig:structure}
\end{minipage}\hfill
\begin{minipage}[t]{0.50\textwidth}
\vspace{0pt}
\small
\captionsetup{justification=raggedright,singlelinecheck=false}
\captionof{table}{
    \textbf{Preprocessing cost on a single A100-SXM4-80GB GPU.}
    All tools were benchmarked on a fixed pool of 50 TCR--pMHC pairs,
    50 TCRs and 20 pMHCs. TCRmodel2 uses 416-residue complexes rather
    than 610-residue inputs, making its reported savings not directly
    comparable. Times include inference only. \Cref{app:cost} gives the per-tool caveats.
}
\label{tab:cost}
\vspace{-0.02in}
\setlength{\tabcolsep}{4.5pt}
\renewcommand{\arraystretch}{2.0}
\renewcommand{\Cell}[1]{{\renewcommand{\arraystretch}{0.95}\begin{tabular}{@{}c@{}}#1\end{tabular}}}
\resizebox{\linewidth}{!}{%
\begin{tabular}{@{}l ccc cc c@{}}
        \toprule
        & \multicolumn{3}{c}{one structure (s)}
        & \multicolumn{2}{c}{this evaluation set (GPU-h)} & \\
        \cmidrule(lr){2-4}
        \cmidrule(lr){5-6}
        tool
        & \Cell{TCR\\228\,aa}
        & \Cell{pMHC\\382\,aa}
        & \Cell{complex\\610\,aa}
        & \Cell{\sys\\21,830 monomers}
        & \Cell{complex\\45,909 pairs}
        & saving \\
        \midrule
        tFold-TCR
        & \textbf{0.672}
        & \textbf{1.618}
        & \textbf{6.329}
        & \textbf{4.2}
        & 80.7
        & 19.4$\times$ \\
        ESMFold2
        & 5.190
        & 10.522
        & 25.629
        & 32.0
        & 326.8
        & 10.2$\times$ \\
        Boltz-2
        & 10.92
        & 17.50
        & 23.784
        & 66.9
        & 303.3
        & 4.5$\times$ \\
        TCRmodel2
        & 102.8
        & 111.5
        & 125.6
        & 624.1
        & 1{,}602.3
        & 2.6$\times$ \\
        \bottomrule
    \end{tabular}%
}

\end{minipage}
\end{figure}

\vspace{-0.15in}
\paragraph{Preprocessing Cost.}
\label{sec:exp:cost}

\sys has a preprocessing cost of $\mathcal{O}(N{+}M)$. All four tools mentioned ran on one A100-SXM4-80GB over one fixed set of samples, and \cref{tab:cost} gives the result. The evaluation set contains 21,477 TCR, 353 pMHC, and 45,909 candidate pairs. \sys reduces preprocessing cost by reusing each monomer structure across the pairs it appears in. Folding every monomer with tFold-TCR costs 4.2 GPU-h. The same set cost 80.7, 303.3, 326.8 and 1{,}602.3 GPU-h as complexes under tFold-TCR, Boltz-2, ESMFold2 and TCRmodel2, a saving of up to 19.4$\times$ against complexes.

\vspace{0.10in}
\section{Conclusion}
\vspace{0.1in}
\label{sec:conclusion}
We propose \syb, a pose-independent framework for TCR--pMHC binding prediction that avoids \textbf{assembly-induced interface corruption (AIC)}, a failure mode in which complex-placement errors alter the residues and contacts presented to downstream predictors. Testing across 49 crystal structures on six folding tools, we show that such placement errors are substantial. By encoding each molecule independently, \sys removes relative pose from the binding score and is first on every metric under all four folding tools, improving MCC at the fixed threshold of 0.5 by \textbf{51 to 75\%} over the strongest baseline and ranking stability across folding tools by \textbf{26 to 85\%} over structural baselines. It also cuts preprocessing by up to \textbf{19.4$\times$} on the same evaluation set. These results suggest a different way of using predicted structure for the TCR--pMHC prediction task. \sys uses only the trustworthy part of TCR and pMHC structures, and it is both more accurate and more stable for it.

\clearpage

\bibliographystyle{plainnat}
\bibliography{refs}

\clearpage

\beginappendix

\label{sec:append}
\part*{Supplementary Material}

{
\setlength{\parskip}{-0em}
\startcontents[sections]
\printcontents[sections]{ }{1}{}
}
\vspace{0.1in}
\section{Notation}
\label{app:notation}
\vspace{0.1in}
Every symbol used by \sys, with the value it takes in the reported configuration.
Indices $i,j$ always run over residues of a single monomer. No symbol in the table ranges over
both molecules at once, which is the notational form of \cref{prop:pose}.
Calligraphic letters denote sets and graphs, uppercase Latin letters denote matrices and
whole-molecule quantities, lowercase Latin letters denote per-residue vectors and scalars, and
Greek letters denote learned networks and fixed nonlinearities. Symbols follow that order within
each block.

\begin{longtable}{@{}l p{0.575\linewidth} l@{}}
\caption{Notation for \sys. Shapes are per residue unless stated otherwise.}
\label{tab:notation}\\
\toprule
symbol & meaning & shape or value \\
\midrule
\endfirsthead
\multicolumn{3}{@{}l}{\textit{\cref{tab:notation} continued}} \\
\toprule
symbol & meaning & shape or value \\
\midrule
\endhead
\bottomrule
\endlastfoot
\multicolumn{3}{@{}l}{\textit{graphs and inputs}} \\
$\mathcal{G}^{T},\mathcal{G}^{P}$ & residue graphs of the receptor and of the antigen, built from separately folded monomers and never merged & --- \\

$\mathcal{N}(i)$ & the $k$ nearest neighbors of $i$ at layer $\ell$ & $|\mathcal{N}(i)|=k$ \\
$N_T,N_P$ & residues in $\mathcal{G}^{T}$ ($\beta$ chain) and in $\mathcal{G}^{P}$ & $\approx114$, $\approx382$ \\
$W_e$ & learned bottleneck that compresses $e_i$ & $\mathbb{R}^{128\times1152}$ \\
$i,j$ & residue indices within one graph & --- \\
$p_i$ & C$\alpha$ coordinate of residue $i$, in that monomer's own frame & $\mathbb{R}^{3}$ \\
$x_i$ & invariant node descriptor: curvature ratios, dihedrals, solvent-accessible area, neighbor count, pLDDT, residue one-hot, secondary-structure one-hot & $\mathbb{R}^{37}$ \\
$e_i$ & frozen per-residue ESM-C embedding, cached once & $\mathbb{R}^{1152}$ \\
$h^{(0)}_i$ & encoder input, $[\,x_i\,\|\,W_e e_i\,]$ & $\mathbb{R}^{165}$ \\
$k$ & neighbors per node & $30$ \\
$d_{ij}$ & C$\alpha$ distance $\lVert p_i-p_j\rVert$ & \AA{} \\
$b_{ij}$ & biochemical compatibility from residue identity alone: hydrogen bonding, hydrophobic product, van der Waals ratio, charge product, aromatic flag & $\mathbb{R}^{5}$ \\
$a_{ij}$ & edge feature, $[\,\rho(d_{ij})\,\|\,d_0/d_{ij}\,\|\,b_{ij}\,]$, $d_0=4$~\AA{} & $\mathbb{R}^{22}$ \\
$\rho$ & radial basis expansion of $d_{ij}$ over $[0,20]$~\AA{} & $\mathbb{R}^{16}$ \\
\multicolumn{3}{@{}l}{\textit{monomer encoder}} \\
$\mathcal{K}^{(\ell)}$ & edge set at layer $\ell$, $\mathrm{kNN}(p^{(\ell)},k)$, rebuilt from the updated coordinates & $kN$ edges \\
$L$ & encoder layers & $2$ \\
$D$ & encoder width & $256$ \\
$H^{T},H^{P}$ & final states $h^{(L)}$ stacked by residue & $\mathbb{R}^{N_T\times D},\mathbb{R}^{N_P\times D}$ \\
$h^{(\ell)}_i$ & invariant node state after layer $\ell$; $h^{(0)}_i$ is projected to width $D$ before layer~1 & $\mathbb{R}^{D}$ \\
$p^{(\ell)}_i$ & coordinate after layer $\ell$; $p^{(0)}_i=p_i$ & $\mathbb{R}^{3}$ \\
$c_{ij}$ & edge context at layer $\ell$, $[\,h^{(\ell-1)}_i\,\|\,h^{(\ell-1)}_j\,\|\,a_{ij}\,]$ & $\mathbb{R}^{2D+22}$ \\
$m_{ij}$ & message carried on edge $(i,j)$ & $\mathbb{R}^{D}$ \\
$g_{ij}$ & chemistry gate $\sigma(\phi_g(b_{ij}))$ scaling that message & $(0,1)$ \\
$\theta$ & encoder parameters, shared by both graphs & --- \\
$\ell$ & layer index, $\ell=1,\dots,L$ & --- \\
$\phi_m$ & message network, two hidden layers with SiLU & $\mathbb{R}^{2D+22}\!\to\!\mathbb{R}^{D}$ \\
$\phi_g$ & gate network on $b_{ij}$, two hidden layers with SiLU & $\mathbb{R}^{5}\!\to\!\mathbb{R}$ \\
$\phi_x$ & coordinate-weight network, $\tanh$-bounded output & $\mathbb{R}^{2D+22}\!\to\!(-1,1)$ \\
$\sigma$ & logistic sigmoid & --- \\
$\mathrm{LN}$ & layer normalization over the feature axis & --- \\
\multicolumn{3}{@{}l}{\textit{contact head and objective}} \\
$\mathcal{B}$ & a training batch of pairs, $\mathcal{B}=\mathcal{B}^{+}\cup\mathcal{B}^{-}$ & --- \\
$\mathcal{B}^{+},\mathcal{B}^{-}$ & the binders and the non-binders of $\mathcal{B}$ & $|\mathcal{B}^{-}|=3|\mathcal{B}^{+}|$ \\
$\mathcal{L}$ & objective, $\mathcal{L}_{\mathrm{BCE}}+\mathcal{L}_{\mathrm{BT}}$ & --- \\
$C$ & rows of $H^{T}$ holding CDR3$\beta$ residues, fixed by sequence annotation & --- \\
$E$ & rows of $H^{P}$ holding peptide residues, fixed by sequence annotation & --- \\
$W_z$ & projection of the concatenated pair & $\mathbb{R}^{128\times256}$ \\
$S$ & slots of the canonical peptide frame & $9$ \\
$d$ & width of one attention head, four heads in all & $32$ \\
$z_t$ & direction~1 output: CDR3 rows query the peptide & $\mathbb{R}^{128}$ \\
$z_p$ & direction~2 output: peptide rows query the CDR3 loops & $\mathbb{R}^{128}$ \\
$z$ & pair embedding, $\mathrm{LN}(W_z[\,z_t\|z_p\,])$ & $\mathbb{R}^{128}$ \\
$w$ & classifier weight & $\mathbb{R}^{128}$ \\
$u$ & pair logit, $w^{\top}z$ & $\mathbb{R}$ \\
$\hat{y}$ & predicted binding probability, $\sigma(u)$ & $(0,1)$ \\
$y$ & binary label, $1$ for a binder and $0$ otherwise & $\{0,1\}$ \\
$\pi$ & monotone map from peptide residues onto the canonical slots & $\{1..|E|\}\!\to\!\{1..S\}$ \\
$\mathrm{MHA}(X,Y)$ & multi-head attention, queries from $X$, keys and values from $Y$ & --- \\
$\mathrm{pool}_{C}$ & mask-aware mean and max over the rows of $C$, concatenated & $\mathbb{R}^{256}$ \\
$\mathrm{pool}_{E}$ & per-slot mean and max under $\pi$, concatenated over the $S$ slots & $\mathbb{R}^{2S\cdot128}$ \\
$\mathrm{proj}_1$ & linear map onto the contact dimension & $\mathbb{R}^{256}\!\to\!\mathbb{R}^{128}$ \\
$\mathrm{proj}_2$ & linear map onto the contact dimension & $\mathbb{R}^{2304}\!\to\!\mathbb{R}^{128}$ \\
\multicolumn{3}{@{}l}{\textit{pose independence}} \\
$R_1,R_2$ & rotations applied independently to the two molecules & $SO(3)$ \\
$t_1,t_2$ & translations applied independently to the two molecules & $\mathbb{R}^{3}$ \\
$g=(R,t)$ & one rigid motion; $g_1$ acts on the receptor and $g_2$ on the antigen & --- \\
$s$ & the score as a function of the two graphs, $s(\mathcal{G}^{T},\mathcal{G}^{P})=\hat{y}$ & $(0,1)$ \\
$N,M$ & distinct receptors and targets in a screen, not residue counts & --- \\
\end{longtable}

\vspace{0.1in}
\section{Proof of Main Text}
\label{app:proof}
\vspace{0.1in}

Write $(R_1,t_1)$ for the motion applied to the receptor and $(R_2,t_2)$ for the one applied to the
antigen. The two are unrelated, which is what separates \cref{prop:pose} from invariance to a single
motion of a joint structure.

\vspace{0.1in}
\subsection{\cref{lem:layer}}
\vspace{0.1in}

\textit{Proof of \cref{lem:layer}.}
\label{app:lemma2}

In \cref{eq:msg} $\phi_m$ receives only the edge context $c_{ij}$ and $\phi_g$ only $b_{ij}$. The
hypothesis leaves $h^{(\ell-1)}$ unchanged, and $a_{ij}$, rebuilt from $p^{(\ell-1)}$, is unchanged by
the argument of \cref{def:inputs} applied to $p^{(\ell-1)}$, so $m_{ij}$ and the gate $g_{ij}$ are unchanged.
Their gated sum followed by $\mathrm{LN}$ gives $h^{(\ell)}$ unchanged. Writing $q_i=p^{(\ell-1)}_i$
and substituting $q\mapsto Rq+t$ into the coordinate update of \cref{eq:upd},
\[
(Rq_i+t)+\frac{1}{|\mathcal{N}(i)|}\sum_{j\in\mathcal{N}(i)}\big(Rq_i-Rq_j\big)\,\phi_x(\cdot)
= R\Big(q_i+\frac{1}{|\mathcal{N}(i)|}\sum_{j\in\mathcal{N}(i)}(q_i-q_j)\,\phi_x(\cdot)\Big)+t
= Rp^{(\ell)}_i+t,
\]
where $\phi_x$ is unchanged because its arguments are, and the translation cancels inside each
difference. The edge set rebuilt from $p^{(\ell)}$ is the same by the argument
of \cref{def:inputs}.

\vspace{0.1in}

\subsection{\cref{prop:pose}}
\vspace{0.1in}

\textit{Proof of \cref{prop:pose}.}
\label{app:prop}

Apply \cref{def:inputs} to each graph, then \cref{lem:layer} by induction over the $L$ layers. So
$H^{T}$ is unchanged by $(R_1,t_1)$ and $H^{P}$ by $(R_2,t_2)$. The two towers share weights but
exchange no message, so neither output depends on the other graph's frame and the two statements
hold simultaneously. The contact head is a function of $H^{T}$, $H^{P}$ and the row masks selecting
$C$ and $E$. It admits no position argument, and the masks are index sets fixed by sequence, so the
score depends on $p$ only through $H$. Hence
$s\big(R_1\mathcal{G}^{T}\!+t_1,\,R_2\mathcal{G}^{P}\!+t_2\big)=s\big(\mathcal{G}^{T},\mathcal{G}^{P}\big)$.

\paragraph{Why the independence matters.}
The guarantee comes from the head's signature rather than from its weights, so it holds at
initialization and cannot be traded away by training. The model has no relative-pose input, so there is nothing for training to ignore. An equivariant network applied to a docked complex, as in
\citet{siriarchawatana2026tcrlens}, gives only the weaker statement that the score is unchanged under
a single motion of the whole assembly. That quotients out the choice of laboratory frame but keeps
the relative pose, which is exactly the quantity we measure to be unreliable. \Cref{prop:pose}
quotients out the relative pose itself, because no operation in the model could read it.

\vspace{0.1in}

\section{Graph Construction and Features}
\label{app:features}
\vspace{0.1in}

\textbf{What is cached.} One monomer is one graph, and a node is a residue represented by its
C$\alpha$. The cache holds the node scalars $x\in\mathbb{R}^{N\times37}$ of \cref{tab:nodefeat} and
the coordinates $p\in\mathbb{R}^{N\times3}$, and nothing else. Edges are not cached, because the
encoder rebuilds the $k$-NN graph and every edge feature from the updated coordinates at each layer.
The per-residue language-model embedding lives in a separate store and its bottleneck is a learned
module rather than a preprocessing step, which is why the cached node vector is $37$ wide rather
than the $165$ the encoder sees.

\textbf{What makes the representation invariant.} Every entry of \cref{tab:nodefeat} and every
biochemical term of \cref{tab:edgefeat} is a function of residue identity or of relative geometry,
so a rigid motion leaves all of them fixed by construction. \Cref{app:proof} turns this into the statement proved for the binding score.

\begin{table}[H]
\centering
\footnotesize\setlength{\tabcolsep}{5pt}
\caption{\textbf{The $37$ node scalars.} Every entry reads residue identity or relative geometry
alone, so none of them introduces an absolute coordinate or a preferred axis. Dimensions $0$--$13$
and $34$--$36$ are the $17$ structure-derived ones that the no-structure ablation zeroes, which
leaves the residue one-hot standing.}
\label{tab:nodefeat}
\begin{tabular}{@{}llp{3.25in}@{}}
\toprule
dims & name & definition \\
\midrule
$0$--$1$   & curvature            & two ratios of the eigenvalues $\lambda_0\!\ge\!\lambda_1\!\ge\!\lambda_2$ of the centered covariance of the $10$ nearest C$\alpha$, namely $\lambda_2/(\lambda_0{+}\lambda_1{+}\lambda_2)$ and $\lambda_1\lambda_2/\lambda_0^{2}$ \\
$2$        & electrostatic        & $\sum_j q_j/(1+d_j)$ over heavy atoms of other residues within $10$~\AA{} of the C$\alpha$, on a per-atom partial-charge table \\
$3$        & SASA                 & Shrake--Rupley, summed over the atoms of the residue \\
$4$        & coordination         & heavy atoms of other residues within $8$~\AA{} \\
$5$--$7$   & backbone dihedral    & $\phi,\psi,\omega$ in radians, undefined values set to $0$ \\
$8$--$11$  & side-chain dihedral  & $\chi_1$ to $\chi_4$ in radians, undefined values set to $0$ \\
$12$       & density              & heavy atoms of other residues within $5$~\AA{} \\
$13$       & pLDDT                & B-factor rescaled to $[0,1]$, with the scale read from the file rather than assumed per tool \\
$14$--$33$ & residue identity     & $20$-way one-hot \\
$34$--$36$ & secondary structure  & $3$-way one-hot from P-SEA, helix, strand and coil \\
\bottomrule
\end{tabular}
\end{table}

\begin{table}[H]
\centering
\footnotesize\setlength{\tabcolsep}{5pt}
\caption{\textbf{The $22$ invariant edge scalars and the five biochemical terms.} Core sits before
secondary on purpose, so the ablation that drops the secondary chemistry is a zeroing of the last
two columns. The order is fixed by the cached graphs and cannot be permuted.}
\label{tab:edgefeat}
\begin{tabular}{@{}llp{3.15in}@{}}
\toprule
cols & name & definition \\
\midrule
$0$--$15$  & radial basis        & $16$ Gaussians with centers spaced evenly over $[0,20]$~\AA{} and $\gamma\approx5.63$ \\
$16$       & radial scalar       & inverse distance $d_0/d_{ij}$ with $d_0=4$~\AA{} \\
$17$--$19$ & core chemistry      & hydrogen bonding, hydrophobic compatibility, van der Waals ratio \\
$20$--$21$ & secondary chemistry & charge product, aromatic pairing \\
\midrule
\multicolumn{3}{@{}l}{\emph{the five biochemical terms, each a function of the two residue identities alone}} \\
\midrule
\multicolumn{2}{@{}l}{hydrogen bonding}        & $(\mathrm{don}_i\wedge\mathrm{acc}_j)\vee(\mathrm{acc}_i\wedge\mathrm{don}_j)$ \\
\multicolumn{2}{@{}l}{hydrophobic compatibility} & $(\mathrm{kd}_i/4.5)\,(\mathrm{kd}_j/4.5)$, a product rather than a difference, so two hydrophobic ends give a positive value and a mismatch a negative one \\
\multicolumn{2}{@{}l}{van der Waals ratio}     & $\min(r_i,r_j)/\max(r_i,r_j)$ with $r=(3V/4\pi)^{1/3}$ on residue volumes \\
\multicolumn{2}{@{}l}{charge product}          & $q_i q_j$ on net residue charge at pH $7.4$ \\
\multicolumn{2}{@{}l}{aromatic pairing}        & $\mathrm{arom}_i\wedge\mathrm{arom}_j$ over PHE, TYR, TRP and HIS \\
\bottomrule
\end{tabular}
\end{table}

\textbf{Why the edge carries an inverse distance.} The edge already holds a charge product
$q_i q_j$ but no radial decay, and the Coulomb kernel $q_i q_j / d$ cannot be written as any one of
the other terms. Left to itself the network has to synthesize $1/d$ out of $16$ local Gaussian bumps,
which are not monotone in $d$ and represent it poorly. A single monotone radial term multiplied by the
charge column is the Coulomb form exactly. Node dimension $2$ of \cref{tab:nodefeat} already uses the
same inverse-distance proxy, and this term is its edge-level counterpart. It is a function of
$\lVert p_i-p_j\rVert$ alone, so it leaves the invariance argument untouched.

\textbf{The hydrogen-bonding sets.} Donors are ARG, ASN, GLN, HIS, LYS, SER, THR, TRP and TYR.
Acceptors are ASN, ASP, GLN, GLU, HIS, SER, THR and TYR. 

\textbf{The invariance gate.} A unit test applies a random rotation and translation to a structure
and requires that the node scalars and the edge scalars are unchanged elementwise. It runs before any bulk build and before any
training run. One source of nondeterminism is known and tolerated. When the $30$th and the $31$st
neighbor of a residue sit within $10^{-3}$~\AA{} of each other, a rotation can swap them and change
the edge set, which we measure on $0.58\%$ of nodes for Boltz-2 structures and $0.86\%$ for tFold-TCR
and ESMFold2, so the gate is allowed to retry.
\vspace{0.1in}

\section{Scope of the Released Model}
\label{app:scope}
\vspace{0.1in}

\textbf{Receptor chains.} The cached receptor graph carries both TCR chains. The encoder is given
the $\beta$ chain alone. Dropping the $\alpha$ chain removes its node block and rebuilds the $k$-NN
edges, so nothing of that chain reaches the encoder. \Cref{app:ablations} reports the arm that
keeps it.

\textbf{Paired $\alpha$ chains.} A source database records an $\alpha$ chain together with its
$\beta$ chain for part of the set only. Matching each $(\mathrm{CDR3}\alpha,\mathrm{CDR3}\beta)$
pair against the paired records of VDJdb, McPAS-TCR and IEDB confirms $6{,}759$ of the $21{,}477$
receptors and $19{,}462$ of the $45{,}909$ rows. The rate is $70.7\%$ on binders and $37.0\%$ on
non-binders. It is a lower bound, because our copy of the 10x Genomics block carries no per-receptor
provenance.

\vspace{0.1in}

\section{AIC Observation Details}
\label{app:folding}
\vspace{0.1in}

\textbf{Targets.} The 49 class-I complexes carry 35 distinct peptides and 31 distinct MHC heavy
chains, so no single system dominates the set. Every tool folds the same five chains taken from
the crystal entry, which makes residue correspondence with the native structure hold by
construction and lets contacts be compared without any alignment step.

\textbf{Leakage control.} Templates are capped per target at the day before that target's
deposition, so no tool can retrieve the answer it is being scored against. As a check that the
cap is sufficient, splitting the 49 by deposition date leaves the possibly-seen half the
\emph{worse} half on docking ($p = 0.010$), the opposite of what memorization would produce.

\begin{table}[H]
\centering
\footnotesize\setlength{\tabcolsep}{3pt}
\caption{\textbf{Interface accuracy of six co-folding tools on the 49 crystal structures}, best per row in bold. Contacts are computed on raw coordinates without superposition, so they are what a downstream graph reads. Every RMSD is measured after superposing the prediction on the MHC heavy chain, except interface RMSD, which follows the CAPRI definition \citep{mendez2005capri} and re-superposes the interface locally. A target counts as well-docked when the receptor lands within 8~\AA{} of its crystal position. Boltz-2 leads almost every row, and \cref{tab:folding-clean} shows why that is not a capability result.}
\label{tab:folding}
\resizebox{\textwidth}{!}{
\begin{tabular}{lcccccc}
\toprule
 & Boltz-2 & ESMFold2 & TCRmodel2 & AF-Multimer & OpenDDE & tFold-TCR \\
\midrule
\multicolumn{7}{l}{\textit{interface contacts, median [mean]}} \\
contact recall @4~\AA{} $\uparrow$ & \textbf{0.87 [0.73]} & 0.52 [0.49] & 0.41 [0.40] & 0.38 [0.40] & 0.39 [0.38] & 0.26 [0.27] \\
contact recall @5~\AA{} $\uparrow$ & \textbf{0.92 [0.78]} & 0.62 [0.58] & 0.52 [0.51] & 0.53 [0.48] & 0.49 [0.44] & 0.30 [0.33] \\
contact recall @8~\AA{} $\uparrow$ & \textbf{0.93 [0.84]} & 0.71 [0.68] & 0.68 [0.62] & 0.64 [0.60] & 0.65 [0.51] & 0.50 [0.49] \\
contact precision @5~\AA{} $\uparrow$ & \textbf{0.89 [0.77]} & 0.62 [0.59] & 0.55 [0.49] & 0.44 [0.45] & 0.40 [0.43] & 0.34 [0.36] \\
contact Jaccard @5~\AA{} $\uparrow$ & \textbf{0.83 [0.69]} & 0.46 [0.44] & 0.34 [0.36] & 0.31 [0.33] & 0.26 [0.33] & 0.19 [0.22] \\
\midrule
\multicolumn{7}{l}{\textit{geometry (\AA), median [mean]}} \\
interface RMSD $\downarrow$ & \textbf{0.43 [1.07]} & 1.80 [2.31] & 2.31 [3.68] & 2.78 [4.02] & 2.65 [5.89] & 3.28 [4.21] \\
TCR docking RMSD $\downarrow$ & \textbf{1.50 [2.80]} & 6.96 [8.12] & 5.87 [8.70] & 8.09 [12.89] & 6.93 [16.21] & 7.14 [9.80] \\
peptide RMSD $\downarrow$ & \textbf{0.43 [0.69]} & 0.72 [1.01] & 0.66 [0.99] & 1.13 [1.70] & 0.88 [1.11] & 1.26 [1.84] \\
CDR3$\beta$ loop RMSD $\downarrow$ & \textbf{0.80 [1.68]} & 2.56 [3.50] & 3.95 [6.17] & 4.13 [5.58] & 4.24 [11.27] & 4.27 [5.86] \\
CDR3$\alpha$ loop RMSD $\downarrow$ & \textbf{0.87 [1.92]} & 3.96 [4.45] & 4.08 [5.60] & 4.07 [5.71] & 4.09 [11.55] & 5.44 [6.29] \\
\midrule
\multicolumn{7}{l}{\textit{receptor placement, fraction of the 49 targets}} \\
docking $<2$~\AA{} (near-native) $\uparrow$ & \textbf{59\%} & 6\% & 8\% & 0\% & 10\% & 4\% \\
docking $<4$~\AA{} $\uparrow$ & \textbf{78\%} & 18\% & 37\% & 14\% & 37\% & 12\% \\
docking $<5$~\AA{} $\uparrow$ & \textbf{80\%} & 29\% & 41\% & 16\% & 39\% & 18\% \\
docking $<8$~\AA{} $\uparrow$ & \textbf{94\%} & 59\% & 78\% & 49\% & 55\% & 57\% \\
docking $>8$~\AA{} (mis-docked) $\downarrow$ & \textbf{6\%} & 41\% & 22\% & 51\% & 45\% & 43\% \\
interface RMSD $<2$~\AA{} (CAPRI medium) $\uparrow$ & \textbf{80\%} & 59\% & 43\% & 33\% & 43\% & 20\% \\
interface RMSD $<4$~\AA{} (CAPRI acceptable) $\uparrow$ & \textbf{96\%} & 90\% & 84\% & 76\% & 67\% & 65\% \\
\midrule
\multicolumn{7}{l}{\textit{contact recall @5~\AA{} split by placement, median}} \\
well-docked targets $\uparrow$ & \textbf{0.92} & 0.72 & 0.63 & 0.64 & 0.69 & 0.42 \\
mis-docked targets $\uparrow$ & 0.20 & \textbf{0.44} & 0.14 & 0.36 & 0.05 & 0.18 \\
\bottomrule
\end{tabular}
}
\end{table}

\begin{table}[H]
\centering
\footnotesize\setlength{\tabcolsep}{3pt}
\caption{\textbf{The same table on the 10 targets released after Boltz-2 training cutoff of 2023-06-01}, the only ones it cannot have memorized. Every tool is scored on the same 10, so the comparison is matched. Boltz-2 goes from first on every contact metric to last on Jaccard and precision. The other five move far less, which is what separates memorization from a harder batch of targets.}
\label{tab:folding-clean}
\resizebox{\textwidth}{!}{
\begin{tabular}{lcccccc}
\toprule
 & Boltz-2 & ESMFold2 & TCRmodel2 & AF-Multimer & OpenDDE & tFold-TCR \\
\midrule
\multicolumn{7}{l}{\textit{interface contacts, median [mean]}} \\
contact recall @4~\AA{} $\uparrow$ & 0.19 [0.25] & \textbf{0.47 [0.49]} & 0.27 [0.31] & 0.39 [0.42] & 0.23 [0.21] & 0.17 [0.20] \\
contact recall @5~\AA{} $\uparrow$ & 0.28 [0.31] & \textbf{0.57 [0.56]} & 0.43 [0.43] & 0.49 [0.50] & 0.30 [0.29] & 0.26 [0.26] \\
contact recall @8~\AA{} $\uparrow$ & 0.45 [0.51] & \textbf{0.70 [0.69]} & 0.64 [0.60] & 0.66 [0.65] & 0.47 [0.41] & 0.41 [0.41] \\
contact precision @5~\AA{} $\uparrow$ & 0.26 [0.30] & \textbf{0.56 [0.61]} & 0.40 [0.40] & 0.45 [0.46] & 0.36 [0.28] & 0.33 [0.33] \\
contact Jaccard @5~\AA{} $\uparrow$ & 0.15 [0.19] & \textbf{0.39 [0.43]} & 0.25 [0.27] & 0.30 [0.34] & 0.20 [0.17] & 0.17 [0.17] \\
\midrule
\multicolumn{7}{l}{\textit{geometry (\AA), median [mean]}} \\
interface RMSD $\downarrow$ & 3.41 [3.24] & \textbf{1.62 [1.89]} & 2.97 [3.41] & 2.29 [3.52] & 3.57 [5.97] & 3.89 [4.28] \\
TCR docking RMSD $\downarrow$ & 6.55 [7.08] & \textbf{6.02 [5.74]} & 6.78 [7.99] & 7.34 [9.83] & 8.35 [16.32] & 9.64 [9.92] \\
peptide RMSD $\downarrow$ & 1.79 [1.56] & \textbf{0.89 [1.07]} & 1.75 [1.54] & 1.98 [1.69] & 1.97 [1.69] & 2.13 [2.07] \\
CDR3$\beta$ loop RMSD $\downarrow$ & 4.91 [4.41] & \textbf{2.44 [2.32]} & 3.96 [5.04] & 3.15 [4.15] & 4.78 [11.72] & 4.23 [4.92] \\
CDR3$\alpha$ loop RMSD $\downarrow$ & 4.55 [5.18] & \textbf{3.10 [3.92]} & 4.99 [5.74] & 3.34 [4.81] & 6.37 [12.99] & 6.87 [6.84] \\
\midrule
\multicolumn{7}{l}{\textit{receptor placement, fraction of the 10 targets}} \\
docking $<2$~\AA{} (near-native) $\uparrow$ & 0\% & \textbf{20\%} & 0\% & 0\% & 0\% & 0\% \\
docking $<4$~\AA{} $\uparrow$ & 10\% & \textbf{40\%} & 20\% & 20\% & 10\% & 20\% \\
docking $<5$~\AA{} $\uparrow$ & 10\% & \textbf{40\%} & 20\% & 20\% & 20\% & 20\% \\
docking $<8$~\AA{} $\uparrow$ & \textbf{70\%} & \textbf{70\%} & \textbf{70\%} & \textbf{70\%} & 40\% & 40\% \\
docking $>8$~\AA{} (mis-docked) $\downarrow$ & \textbf{30\%} & \textbf{30\%} & \textbf{30\%} & \textbf{30\%} & 60\% & 60\% \\
interface RMSD $<2$~\AA{} (CAPRI medium) $\uparrow$ & 10\% & \textbf{60\%} & 20\% & 50\% & 0\% & 0\% \\
interface RMSD $<4$~\AA{} (CAPRI acceptable) $\uparrow$ & 80\% & \textbf{100\%} & 90\% & 90\% & 70\% & 50\% \\
\midrule
\multicolumn{7}{l}{\textit{contact recall @5~\AA{} split by placement, median}} \\
well-docked targets $\uparrow$ & 0.40 & \textbf{0.72} & 0.47 & 0.63 & 0.38 & 0.30 \\
mis-docked targets $\uparrow$ & 0.20 & 0.34 & \textbf{0.42} & 0.40 & 0.27 & 0.21 \\
\bottomrule
\end{tabular}
}
\end{table}

\textbf{Sampling.} \Cref{tab:sampling} gives the per-tool figures behind the coherence argument.
Two of its rows need a caveat. Boltz-2's $0.754$~\AA{} cannot be read as accuracy, because $39$ of
the $49$ targets were released before its training cutoff. Split on that cutoff, the median is
$0.68$~\AA{} inside the training window against $1.72$~\AA{} outside it, Mann--Whitney
$p=7.6\times10^{-5}$, and no other tool shows the split. AF-Multimer needs no such split, since its
training cutoff precedes every target in the set, and restricting it to the ten clean targets leaves
$1.589$~\AA{} and $\gamma=0.969$. Its confidence column is also not on the same footing as the
others, because AlphaFold publishes no per-chain-pair value and the global score is used in its
place, so that one row cannot be compared against the others on calibration.

\begin{table}[H]
\centering
\footnotesize\setlength{\tabcolsep}{5pt}
\caption{\textbf{Twenty co-folds per target do not average AIC away.} Each stochastic tool folds
all $49$ targets $20$ times, which is $980$ structures per tool. AIC is the median over targets of
the per-target median. $\gamma$ is the directional coherence of the $20$ samples, against $0.21$ for uniformly
random directions. The change column is the median over targets of $1-\text{equal}/\text{single}$,
positive meaning better, with a two-sided Wilcoxon signed-rank $p$. The oracle column picks the best
of the $20$ per target and is not available to any deployed system.}
\label{tab:sampling}
\begin{tabular}{@{}lccccc@{}}
\toprule
tool & AIC (\AA{}) & $\gamma$ & equal-weight change & Wilcoxon $p$ & oracle \\
\midrule
Boltz-2      & \textbf{0.754} & 0.949          & $+3.6\%$ & 0.194 & $+44.4\%$ \\
ESMFold2     & 2.227          & 0.839          & $+6.7\%$ & 0.328 & $+52.0\%$ \\
AF-Multimer  & 2.965          & \textbf{0.992} & $+1.3\%$ & 0.851 & $+18.6\%$ \\
TCRmodel2    & 3.062          & 0.985          & $-2.5\%$ & 0.078 & $+27.3\%$ \\
OpenDDE      & 3.961          & 0.946          & $+0.6\%$ & 0.465 & $+34.6\%$ \\
tFold-TCR    & 3.52           & \multicolumn{4}{c}{deterministic, no distribution to sample} \\
\bottomrule
\end{tabular}
\end{table}

\vspace{0.1in}

\section{Experiment Results}
\label{app:results}
\vspace{0.1in}

\subsection{Evaluation Protocol}
\label{app:protocol}
\vspace{0.1in}

\textbf{Composition.} The curated evaluation set contains 7,371 measured binding pairs over 343 epitopes, 7,142 distinct TCRs and 353 class-I pMHCs. Sampled non-binders extend it to 45,909 rows at a prevalence of $0.161$, which is the chance level for every precision metric in the paper. The 343 epitopes are presented on 353 pMHC, because ten of them appear with two MHC alleles. The distribution is long tailed, and the ten most frequent epitopes account for most of the entries. That is a property of the field's data, not of our protocol, and it is the reason macro and micro aggregations are reported side by side.

\textbf{Negative samples.} The non-binders come with the STAG release and are used without
resampling \citep{slone2025stag}. STAG builds them two ways. It takes non-binding pairs from the
10x Genomics dextramer panel \citep{tenx2019immunity}, and it swaps TCRs, holding the pMHC of a
recorded binder fixed and drawing a TCR from the repertoire to pair with it. The released file
carries no provenance column, so the two blocks cannot be separated. \Cref{tab:negatives} gives
the composition. Evaluation scores every row. Training draws three non-binders per binder from the
same pool and redraws them each epoch (\cref{tab:hparams}).

\begin{table}[H]
\centering
\small\setlength{\tabcolsep}{8pt}
\caption{Composition of the $38{,}538$ non-binders. Quartiles run over the $353$ pMHC of the set.}
\label{tab:negatives}
\begin{tabular}{lc}
\toprule
quantity & value \\
\midrule
non-binders per binder within one pMHC, median & 5.5 \\
\quad quartiles                                & 4.9, 6.0 \\
\quad maximum                                  & 32 \\
pMHC carrying at least one binder              & 353 \\
pMHC carrying at least one non-binder          & 348 \\
\addlinespace
distinct TCRs                                  & 21,477 \\
\quad in a binding pair                        & 7,142 \\
\quad in a non-binding pair                    & 19,174 \\
\quad in non-binding pairs only                & 14,335 \\
non-binding rows carried by those TCRs         & 22,314 \\
\addlinespace
times a TCR is reused across non-binders, median & 2 \\
\quad 90th percentile                          & 4 \\
\quad maximum                                  & 26 \\
pairs labeled both binding and non-binding     & 0 \\
\bottomrule
\end{tabular}
\end{table}

\textbf{Fold sizes.} Every epitope forms its own cluster and each cluster is assigned to exactly one fold. No epitope appears in more than one fold.

\begin{table}[H]
\centering
\small\setlength{\tabcolsep}{8pt}
\caption{The five epitope-disjoint folds.}
\label{tab:folds}
\begin{tabular}{lccccc}
\toprule
fold & rows & positives & negatives & pos\,:\,neg & epitopes \\
\midrule
0 & 12,327 & 1,876 & 10,451 & 1 : 5.6 & 65 \\
1 &  9,342 & 1,593 &  7,749 & 1 : 4.9 & 67 \\
2 &  9,337 & 1,531 &  7,806 & 1 : 5.1 & 69 \\
3 &  8,543 & 1,329 &  7,214 & 1 : 5.4 & 70 \\
4 &  6,360 & 1,042 &  5,318 & 1 : 5.1 & 72 \\
\midrule
total & 45,909 & 7,371 & 38,538 & 1 : 5.2 & 343 \\
\bottomrule
\end{tabular}
\end{table}

\textbf{Rotation.} For test fold $F$ the validation fold is $F{+}1 \bmod 5$ and the remaining three train. Checkpoint selection is on validation micro-AUROC, a three-epoch rolling mean with no checkpoint before epoch 10. The test fold is consulted neither in training nor in selection.

\textbf{Per fold, then averaged.} Every metric in the paper is computed within a fold and the five values are averaged. None is computed on the five folds pooled. Each fold is a separately trained model on its own logit scale, so pooling lets between-model score offsets drive the AUROC and silently weights by fold size. An earlier version of this analysis pooled, and reported one configuration at micro-AUROC 0.676 against its true 0.642.

\textbf{Row coverage.} Each method is scored on every row it covers. Coverage differs between methods by under 0.5\% of the fold, because a few structures are missing under one tool and the sequence baselines drop rows they cannot parse. \sys emits a logit and the baselines emit probabilities. Only the \sys side is passed through a sigmoid, which is monotone, so AUROC and AUPRC are untouched and MCC@0.5 on the probability equals MCC@0 on the logit.

\textbf{Degenerate folds.} Every fold of every seed is kept in \cref{tab:main}, including the ones on which a baseline fails. SageTCR collapses to a constant logit on fold~0 under all four structure sources, with a score standard deviation of zero and as few as one distinct value across 12{,}318 candidate pairs, so its AUROC there is 0.5 and its MCC at a fixed threshold is 0. TRAP does not go fully constant but its scores span a range too narrow to cross the 0.5 threshold on folds~0 and~4, which also drives MCC@0.5 to 0. Dropping those folds would report each baseline on the folds it happens to survive, so we keep them and report the mean over all five.

\vspace{0.1in}

\subsection{Metric Definitions}
\label{app:metrics}
\label{app:mcc}
\vspace{0.1in}

\textbf{Ranking.} For a scored set $S$, $\mathrm{AUROC}(S)=\Pr(\hat{y}^{+}\!>\!\hat{y}^{-})$ over random positive--negative pairs and $\mathrm{AUPRC}(S)=\int_{0}^{1}\mathrm{Prec}(r)\,\mathrm{d}r$. Writing $\mathcal{E}$ for the epitopes of a fold and $S_e$ for the candidates of epitope $e$,
\[
\text{micro-}M=M\Big(\bigcup\nolimits_{e\in\mathcal{E}}S_e\Big),\qquad
\text{macro-}M=\frac{1}{|\mathcal{E}|}\sum_{e\in\mathcal{E}}M(S_e),\qquad
M\in\{\mathrm{AUROC},\mathrm{AUPRC}\},
\]
so micro is dominated by the frequent epitopes and macro counts every epitope once. A precision metric starts at the class prevalence rather than at zero, which is $0.161$ here, so a loss is read against the margin a method holds above that level.

\textbf{Classification.} At a threshold $t$ a pair counts as predicted positive when $\hat{y}\ge t$. Writing $\mathrm{TP}$, $\mathrm{TN}$, $\mathrm{FP}$ and $\mathrm{FN}$ for the four counts this produces,
\[
\mathrm{MCC}=\frac{\mathrm{TP}\cdot\mathrm{TN}-\mathrm{FP}\cdot\mathrm{FN}}
{\sqrt{(\mathrm{TP}+\mathrm{FP})(\mathrm{TP}+\mathrm{FN})(\mathrm{TN}+\mathrm{FP})(\mathrm{TN}+\mathrm{FN})}},
\qquad
\mathrm{MCC}@0.5=\mathrm{MCC}\big|_{t=0.5}.
\]
MCC runs over $[-1,1]$ and is $0$ for a constant or random prediction, so it is not inflated by the class imbalance the way accuracy is. We fix $t=0.5$ for every method, which is the cut a system runs at when no fold is held out for calibration. We also report $\mathrm{MCC}^{\ast}=\max_{t}\mathrm{MCC}(t)$ taken on the test fold itself. No deployed system has that threshold, so $\mathrm{MCC}^{\ast}$ is an oracle value and we read it only as a bound on what threshold tuning could buy a baseline. We also report $\mathrm{MCC}@t$, where $t$ is fitted on the four folds that are not being scored. It is a third way of choosing the threshold, and it reads nothing from the fold being scored. \Cref{tab:threshold} gives all three side by side, with the calibration and the operating point each one produces.

\textbf{Stability.} Stability is measured on rankings rather than scores. For two rankings of the same $n$ candidates,
\[\tau_b=\frac{C-D}{\sqrt{(C+D+T_x)(C+D+T_y)}},\qquad
\text{order held}=\frac{1+\tau_b}{2}\]
with $C$ and $D$ the concordant and discordant pairs and $T_x,T_y$ those tied in one ranking only \citep{kendall1938tau, kendall1945ties}. Across all $m=4$ rankings at once we use Kendall's concordance $W=12S/\big(m^{2}(n^{3}-n)\big)$ \citep{kendall1939mrankings}, where $S$ sums the squared deviations of each candidate's rank sum from the mean. $W=1$ is perfect agreement and $W=0$ is none. Overlap@$k$ is the intersection of two top-$k$ sets divided by $k$. Intervals are bootstrapped over epitope-folds.

\begin{table}[H]
\centering
\footnotesize\setlength{\tabcolsep}{3.4pt}
\caption{\textbf{Three ways to set the decision threshold, with the calibration and the operating
point each one produces.} $\mathrm{MCC}@0.5$ uses the fixed cut. $\mathrm{MCC}@t$ uses a threshold
fitted on the four folds that are not being scored, so it reads neither the fixed cut nor the scored
fold.
$\mathrm{MCC}^{\ast}$ uses the threshold that maximizes MCC on the scored fold, which is an oracle.
The first and third columns repeat \cref{tab:main}. Precision, recall and predicted-positive rate
are read at $t$. Brier and ECE are threshold free, and ECE uses ten equal-mass bins. $t$ is the
median of the four per-fold thresholds. Means over five folds and three seeds. Best
$\mathrm{MCC}@t$ per group in bold.}
\label{tab:threshold}
\resizebox{\textwidth}{!}{%
\begin{tabular}{llccccccccc}
\toprule
structures & method & MCC@0.5 & MCC@$t$ & MCC$^{\ast}$ & precision & recall & pred.\ pos. & Brier & ECE & $t$ \\
\midrule
\multirow{5}{*}{tFold-TCR} & \syb & 0.251 & \textbf{0.233} & 0.292 & 0.517 & 0.206 & 0.068 & 0.126 & 0.070 & 0.452 \\
 & STAG-LLM & 0.166 & 0.170 & 0.264 & 0.442 & 0.132 & 0.049 & 0.188 & 0.216 & 0.720 \\
 & SageTCR & 0.122 & 0.137 & 0.182 & 0.399 & 0.101 & 0.039 & 0.133 & 0.040 & 0.318 \\
 & Abeer & 0.071 & 0.108 & 0.157 & 0.293 & 0.120 & 0.058 & 0.237 & 0.303 & 0.745 \\
 & TRAP & 0.034 & 0.046 & 0.095 & 0.409 & 0.015 & 0.006 & 0.200 & 0.226 & 0.639 \\
\addlinespace[2pt]\cmidrule(lr){1-11}\addlinespace[2pt]
\multirow{5}{*}{ESMFold2} & \syb & 0.273 & \textbf{0.257} & 0.304 & 0.505 & 0.253 & 0.086 & 0.126 & 0.077 & 0.435 \\
 & STAG-LLM & 0.172 & 0.224 & 0.245 & 0.465 & 0.224 & 0.079 & 0.192 & 0.210 & 0.748 \\
 & SageTCR & 0.115 & 0.117 & 0.174 & 0.336 & 0.105 & 0.049 & 0.135 & 0.041 & 0.298 \\
 & Abeer & 0.108 & 0.110 & 0.188 & 0.297 & 0.115 & 0.054 & 0.231 & 0.300 & 0.706 \\
 & TRAP & 0.041 & 0.061 & 0.096 & 0.402 & 0.019 & 0.007 & 0.196 & 0.224 & 0.621 \\
\addlinespace[2pt]\cmidrule(lr){1-11}\addlinespace[2pt]
\multirow{5}{*}{TCRmodel2} & \syb & 0.267 & \textbf{0.209} & 0.301 & 0.535 & 0.180 & 0.060 & 0.127 & 0.082 & 0.487 \\
 & STAG-LLM & 0.153 & 0.169 & 0.264 & 0.422 & 0.164 & 0.066 & 0.200 & 0.237 & 0.736 \\
 & SageTCR & 0.112 & 0.128 & 0.174 & 0.338 & 0.119 & 0.052 & 0.134 & 0.043 & 0.248 \\
 & Abeer & 0.145 & 0.123 & 0.236 & 0.329 & 0.134 & 0.066 & 0.224 & 0.289 & 0.701 \\
 & TRAP & 0.019 & 0.035 & 0.077 & 0.253 & 0.015 & 0.007 & 0.201 & 0.224 & 0.641 \\
\addlinespace[2pt]\cmidrule(lr){1-11}\addlinespace[2pt]
\multirow{5}{*}{Boltz-2} & \syb & 0.278 & \textbf{0.263} & 0.311 & 0.623 & 0.192 & 0.048 & 0.127 & 0.079 & 0.559 \\
 & STAG-LLM & 0.169 & 0.183 & 0.254 & 0.454 & 0.160 & 0.058 & 0.198 & 0.218 & 0.800 \\
 & SageTCR & 0.115 & 0.114 & 0.146 & 0.310 & 0.122 & 0.060 & 0.134 & 0.046 & 0.257 \\
 & Abeer & 0.054 & 0.090 & 0.188 & 0.251 & 0.145 & 0.089 & 0.251 & 0.319 & 0.683 \\
 & TRAP & 0.016 & 0.042 & 0.082 & 0.270 & 0.014 & 0.006 & 0.192 & 0.206 & 0.661 \\
\midrule
\multirow{2}{*}{-} & NetTCR-2.2 & 0.207 & 0.241 & 0.257 & 0.493 & 0.212 & 0.065 & 0.128 & 0.060 & 0.360 \\
 & ERGO-II & 0.061 & 0.199 & 0.224 & 0.474 & 0.161 & 0.056 & 0.282 & 0.373 & 0.799 \\
\bottomrule
\end{tabular}%
}
\end{table}

\subsection{Hyperparameter}
\label{app:training}
\vspace{0.1in}

The objective sums a cross-entropy term and a Bradley--Terry ranking term, both defined in the method. Negatives are redrawn every epoch rather than fixed once, and checkpoints are selected on a rolling three-epoch mean of validation micro-AUROC, because taking the single best epoch selects on the validation fold.

\begin{table}[H]
\centering
\small\setlength{\tabcolsep}{7pt}
\caption{Every hyperparameter of the reported \sys runs. Values are shared by all arms. The arms differ only in the flag named in \cref{app:ablations}.}
\label{tab:hparams}
\begin{tabular}{lll}
\toprule
group & setting & value \\
\midrule
graph      & node descriptor $x$                & $37$ invariant scalars \\
           & language model                     & ESM-C 600M, frozen, $1152$-d \\
           & learned bottleneck $W_e$           & $1152\rightarrow128$ \\
           & encoder input $h^{(0)}$            & $165$-d \\
           & neighbors $k$                     & $30$, rebuilt every layer \\
           & edge feature $a_{ij}$              & $16$ RBF over $[0,20]$~\AA{} $+$ $d_0/d_{ij}$ ($d_0=4$~\AA{}) $+$ $5$ biochemical \\
\midrule
encoder    & layers $L$                         & $2$ \\
           & width                              & $256$ \\
           & activation                         & SiLU \\
           & coordinate update                  & Tanh-bounded coefficient \\
\midrule
head       & attention heads                    & $4\times32$ \\
           & contact dimension                  & $128$ \\
           & pooling                            & mask-aware mean $+$ max \\
\midrule
optimiser  & algorithm                          & AdamW \citep{loshchilov2019adamw} \\
           & learning rate                      & $3\times10^{-5}$ \\
           & weight decay                       & $0.01$ \\
           & schedule                           & $5$ warm-up epochs, then cosine \\
           & epochs                             & $100$ \\
\midrule
sampling   & negatives per positive             & $3$, redrawn every epoch \\
           & positive sampling                  & uniform at natural epitope frequency \\
\midrule
selection  & criterion                          & rolling mean (window $3$) of validation micro-AUROC \\
           & floor                              & epoch $10$ \\
\bottomrule
\end{tabular}
\end{table}

\subsection{Binding Prediction, Fold by Fold}
\label{app:perfold}

Every mean in \cref{tab:main} is an average of five per-fold values, and the tables below
give all of them. The spread across folds is wide enough to be worth stating plainly: fold~0 is the easiest fold on the micro
metrics for every method that separates the classes at all, and the three weakest methods peak
elsewhere only because they barely rank on fold~0 either. On the macro metrics, which average
over epitopes and so discount that one large epitope, no method peaks on fold~0.

The last column of each table is MCC at the threshold that maximizes it on that fold, the
oracle value of \cref{tab:threshold}, so the cost of the fixed cut can be read fold by fold.

\textbf{Fold by fold against the strongest baseline.} The means of \cref{tab:main} say which
method is ahead on average. \Cref{tab:paired} asks the same question one fold at a time, pairing on
the fold because that is the only unit two methods share. Seeds are not a shared unit, since each
method draws its own, so seeds are averaged within a fold before pairing. \sys wins all five folds
in 16 of the 24 cells. The exceptions are micro-AUROC under ESMFold2 and Boltz-2, where the median
difference is $+0.005$ and $+0.009$ and three folds go its way, micro-AUPRC under tFold-TCR and
ESMFold2, where four folds do, and $\mathrm{MCC}^{\ast}$, where no tool reaches five. Both exceptions sit where the main text already reads the lead as small, on the
whole-list metric and after the threshold has been tuned on the test fold.

\begin{table}[H]
\centering
\footnotesize\setlength{\tabcolsep}{6pt}
\caption{\textbf{Fold-by-fold comparison against the strongest baseline in each cell.} Seeds are
averaged within a fold first, so each cell is five paired observations, one per epitope-disjoint
fold. The entry gives the folds \sys wins and the median paired difference. Five folds out of five
is a one-sided Wilcoxon $p$ of $0.031$, which is the smallest value five paired observations can
produce, so no cell can report a stronger test. The baseline is STAG-LLM except for macro-AUPRC
under tFold-TCR, ESMFold2 and TCRmodel2 and macro-AUROC under tFold-TCR, which are against TRAP,
and micro-AUROC under TCRmodel2, which is against Abeer.}
\label{tab:paired}
\begin{tabular}{lcccc}
\toprule
metric & tFold-TCR & ESMFold2 & TCRmodel2 & Boltz-2 \\
\midrule
macro-AUROC & \textbf{5/5}\,\small{(+0.100)} & \textbf{5/5}\,\small{(+0.045)} & \textbf{5/5}\,\small{(+0.130)} & \textbf{5/5}\,\small{(+0.036)} \\
macro-AUPRC & \textbf{5/5}\,\small{(+0.067)} & \textbf{5/5}\,\small{(+0.060)} & \textbf{5/5}\,\small{(+0.131)} & \textbf{5/5}\,\small{(+0.064)} \\
micro-AUROC & \textbf{5/5}\,\small{(+0.062)} & 3/5\,\small{(+0.005)} & \textbf{5/5}\,\small{(+0.075)} & 3/5\,\small{(+0.009)} \\
micro-AUPRC & 4/5\,\small{(+0.060)} & 4/5\,\small{(+0.031)} & \textbf{5/5}\,\small{(+0.058)} & \textbf{5/5}\,\small{(+0.022)} \\
MCC@0.5 & \textbf{5/5}\,\small{(+0.082)} & \textbf{5/5}\,\small{(+0.048)} & \textbf{5/5}\,\small{(+0.075)} & \textbf{5/5}\,\small{(+0.068)} \\
MCC$^{\ast}$ & 3/5\,\small{(+0.017)} & 4/5\,\small{(+0.026)} & 4/5\,\small{(+0.017)} & 4/5\,\small{(+0.025)} \\
\bottomrule
\end{tabular}
\end{table}

\begin{table}[H]
\centering
\footnotesize\setlength{\tabcolsep}{2.2pt}
\caption{\textbf{Per-fold results, structures from tFold-TCR}. Each cell is the mean over three training seeds, $\pm$ the spread across those seeds. Every \emph{mean} row averages the five folds within a seed first, so it matches \cref{tab:main}. No metric is computed on the five folds pooled. Candidate counts differ between methods by under 0.5\%, so they are not repeated here.}
\label{tab:perfold-tfold}
\begin{tabular}{lcccccc}
\toprule
method / fold & macro-AUROC & macro-AUPRC & micro-AUROC & micro-AUPRC & MCC@0.5 & MCC$^{\ast}$ \\
\midrule
\multicolumn{7}{l}{\textit{\syb}} \\
\quad 0 & 0.713\,\sd{0.014} & 0.530\,\sd{0.023} & 0.855\,\sd{0.010} & 0.656\,\sd{0.049} & 0.512\,\sd{0.279} & 0.648\,\sd{0.045} \\
\quad 1 & 0.776\,\sd{0.024} & 0.611\,\sd{0.025} & 0.628\,\sd{0.005} & 0.278\,\sd{0.008} & 0.144\,\sd{0.010} & 0.154\,\sd{0.004} \\
\quad 2 & 0.788\,\sd{0.010} & 0.615\,\sd{0.002} & 0.582\,\sd{0.012} & 0.283\,\sd{0.015} & 0.172\,\sd{0.020} & 0.208\,\sd{0.023} \\
\quad 3 & 0.777\,\sd{0.014} & 0.613\,\sd{0.018} & 0.625\,\sd{0.013} & 0.273\,\sd{0.024} & 0.152\,\sd{0.033} & 0.169\,\sd{0.031} \\
\quad 4 & 0.736\,\sd{0.008} & 0.555\,\sd{0.004} & 0.639\,\sd{0.002} & 0.343\,\sd{0.004} & 0.275\,\sd{0.016} & 0.279\,\sd{0.014} \\
\cmidrule(l){1-7}
\quad \textit{mean} & 0.758\,\sd{0.010} & 0.585\,\sd{0.008} & 0.666\,\sd{0.003} & 0.366\,\sd{0.014} & 0.251\,\sd{0.059} & 0.292\,\sd{0.017} \\
\addlinespace[2pt]\midrule
\multicolumn{7}{l}{\textit{STAG-LLM}} \\
\quad 0 & 0.620\,\sd{0.021} & 0.480\,\sd{0.036} & 0.792\,\sd{0.028} & 0.455\,\sd{0.061} & 0.445\,\sd{0.107} & 0.474\,\sd{0.116} \\
\quad 1 & 0.642\,\sd{0.058} & 0.475\,\sd{0.067} & 0.490\,\sd{0.015} & 0.242\,\sd{0.026} & 0.036\,\sd{0.037} & 0.192\,\sd{0.033} \\
\quad 2 & 0.649\,\sd{0.027} & 0.496\,\sd{0.029} & 0.532\,\sd{0.005} & 0.223\,\sd{0.008} & 0.090\,\sd{0.033} & 0.171\,\sd{0.016} \\
\quad 3 & 0.642\,\sd{0.015} & 0.500\,\sd{0.014} & 0.574\,\sd{0.002} & 0.275\,\sd{0.020} & 0.127\,\sd{0.031} & 0.221\,\sd{0.015} \\
\quad 4 & 0.546\,\sd{0.031} & 0.411\,\sd{0.022} & 0.522\,\sd{0.021} & 0.264\,\sd{0.002} & 0.134\,\sd{0.114} & 0.263\,\sd{0.003} \\
\cmidrule(l){1-7}
\quad \textit{mean} & 0.620\,\sd{0.018} & 0.473\,\sd{0.020} & 0.582\,\sd{0.008} & 0.292\,\sd{0.006} & 0.166\,\sd{0.024} & 0.264\,\sd{0.018} \\
\addlinespace[2pt]\midrule
\multicolumn{7}{l}{\textit{SageTCR}} \\
\quad 0 & 0.512\,\sd{0.016} & 0.243\,\sd{0.079} & 0.519\,\sd{0.028} & 0.156\,\sd{0.006} & 0.000\,\sd{0.000} & 0.120\,\sd{0.113} \\
\quad 1 & 0.643\,\sd{0.027} & 0.474\,\sd{0.020} & 0.510\,\sd{0.017} & 0.251\,\sd{0.019} & 0.176\,\sd{0.047} & 0.199\,\sd{0.023} \\
\quad 2 & 0.668\,\sd{0.008} & 0.511\,\sd{0.030} & 0.554\,\sd{0.022} & 0.253\,\sd{0.021} & 0.128\,\sd{0.062} & 0.189\,\sd{0.025} \\
\quad 3 & 0.671\,\sd{0.018} & 0.525\,\sd{0.027} & 0.568\,\sd{0.016} & 0.260\,\sd{0.029} & 0.154\,\sd{0.069} & 0.180\,\sd{0.041} \\
\quad 4 & 0.615\,\sd{0.045} & 0.459\,\sd{0.020} & 0.584\,\sd{0.018} & 0.282\,\sd{0.012} & 0.149\,\sd{0.004} & 0.223\,\sd{0.005} \\
\cmidrule(l){1-7}
\quad \textit{mean} & 0.622\,\sd{0.009} & 0.443\,\sd{0.004} & 0.547\,\sd{0.011} & 0.240\,\sd{0.007} & 0.122\,\sd{0.027} & 0.182\,\sd{0.036} \\
\addlinespace[2pt]\midrule
\multicolumn{7}{l}{\textit{Abeer}} \\
\quad 0 & 0.492\,\sd{0.049} & 0.341\,\sd{0.036} & 0.601\,\sd{0.085} & 0.202\,\sd{0.065} & 0.087\,\sd{0.096} & 0.130\,\sd{0.094} \\
\quad 1 & 0.618\,\sd{0.065} & 0.455\,\sd{0.063} & 0.494\,\sd{0.030} & 0.225\,\sd{0.021} & 0.032\,\sd{0.060} & 0.158\,\sd{0.039} \\
\quad 2 & 0.672\,\sd{0.016} & 0.506\,\sd{0.021} & 0.546\,\sd{0.009} & 0.221\,\sd{0.008} & 0.072\,\sd{0.016} & 0.143\,\sd{0.016} \\
\quad 3 & 0.613\,\sd{0.028} & 0.465\,\sd{0.031} & 0.558\,\sd{0.036} & 0.205\,\sd{0.027} & 0.065\,\sd{0.017} & 0.124\,\sd{0.056} \\
\quad 4 & 0.557\,\sd{0.007} & 0.402\,\sd{0.008} & 0.550\,\sd{0.043} & 0.274\,\sd{0.018} & 0.100\,\sd{0.050} & 0.230\,\sd{0.020} \\
\cmidrule(l){1-7}
\quad \textit{mean} & 0.591\,\sd{0.008} & 0.434\,\sd{0.006} & 0.550\,\sd{0.011} & 0.225\,\sd{0.014} & 0.071\,\sd{0.007} & 0.157\,\sd{0.024} \\
\addlinespace[2pt]\midrule
\multicolumn{7}{l}{\textit{TRAP}} \\
\quad 0 & 0.573\,\sd{0.043} & 0.419\,\sd{0.058} & 0.552\,\sd{0.087} & 0.187\,\sd{0.058} & 0.000\,\sd{0.000} & 0.110\,\sd{0.094} \\
\quad 1 & 0.682\,\sd{0.010} & 0.554\,\sd{0.026} & 0.423\,\sd{0.010} & 0.170\,\sd{0.005} & 0.058\,\sd{0.018} & 0.101\,\sd{0.010} \\
\quad 2 & 0.703\,\sd{0.030} & 0.577\,\sd{0.020} & 0.511\,\sd{0.022} & 0.191\,\sd{0.010} & 0.038\,\sd{0.009} & 0.094\,\sd{0.006} \\
\quad 3 & 0.662\,\sd{0.029} & 0.525\,\sd{0.051} & 0.515\,\sd{0.014} & 0.174\,\sd{0.012} & 0.043\,\sd{0.027} & 0.084\,\sd{0.029} \\
\quad 4 & 0.635\,\sd{0.045} & 0.488\,\sd{0.055} & 0.494\,\sd{0.006} & 0.188\,\sd{0.019} & 0.029\,\sd{0.050} & 0.088\,\sd{0.051} \\
\cmidrule(l){1-7}
\quad \textit{mean} & 0.651\,\sd{0.015} & 0.513\,\sd{0.014} & 0.499\,\sd{0.010} & 0.182\,\sd{0.010} & 0.034\,\sd{0.000} & 0.095\,\sd{0.024} \\
\bottomrule
\end{tabular}
\end{table}

\begin{table}[H]
\centering
\footnotesize\setlength{\tabcolsep}{2.2pt}
\caption{\textbf{Per-fold results, structures from ESMFold2}. Conventions as in \cref{tab:perfold-tfold}.}
\label{tab:perfold-esmfold2}
\begin{tabular}{lcccccc}
\toprule
method / fold & macro-AUROC & macro-AUPRC & micro-AUROC & micro-AUPRC & MCC@0.5 & MCC$^{\ast}$ \\
\midrule
\multicolumn{7}{l}{\textit{\syb}} \\
\quad 0 & 0.737\,\sd{0.015} & 0.563\,\sd{0.020} & 0.855\,\sd{0.013} & 0.677\,\sd{0.045} & 0.675\,\sd{0.051} & 0.690\,\sd{0.036} \\
\quad 1 & 0.774\,\sd{0.014} & 0.605\,\sd{0.008} & 0.660\,\sd{0.003} & 0.324\,\sd{0.017} & 0.161\,\sd{0.041} & 0.206\,\sd{0.014} \\
\quad 2 & 0.773\,\sd{0.006} & 0.587\,\sd{0.011} & 0.577\,\sd{0.005} & 0.263\,\sd{0.013} & 0.165\,\sd{0.003} & 0.184\,\sd{0.012} \\
\quad 3 & 0.781\,\sd{0.016} & 0.617\,\sd{0.013} & 0.621\,\sd{0.012} & 0.262\,\sd{0.017} & 0.137\,\sd{0.041} & 0.174\,\sd{0.011} \\
\quad 4 & 0.729\,\sd{0.004} & 0.545\,\sd{0.006} & 0.639\,\sd{0.005} & 0.334\,\sd{0.007} & 0.228\,\sd{0.022} & 0.267\,\sd{0.015} \\
\cmidrule(l){1-7}
\quad \textit{mean} & 0.759\,\sd{0.004} & 0.584\,\sd{0.003} & 0.670\,\sd{0.005} & 0.372\,\sd{0.013} & 0.273\,\sd{0.029} & 0.304\,\sd{0.011} \\
\addlinespace[2pt]\midrule
\multicolumn{7}{l}{\textit{STAG-LLM}} \\
\quad 0 & 0.705\,\sd{0.010} & 0.489\,\sd{0.010} & 0.819\,\sd{0.017} & 0.463\,\sd{0.052} & 0.360\,\sd{0.031} & 0.472\,\sd{0.049} \\
\quad 1 & 0.758\,\sd{0.013} & 0.550\,\sd{0.025} & 0.608\,\sd{0.024} & 0.293\,\sd{0.016} & 0.131\,\sd{0.003} & 0.190\,\sd{0.025} \\
\quad 2 & 0.728\,\sd{0.006} & 0.543\,\sd{0.009} & 0.603\,\sd{0.012} & 0.237\,\sd{0.009} & 0.117\,\sd{0.020} & 0.126\,\sd{0.019} \\
\quad 3 & 0.714\,\sd{0.024} & 0.523\,\sd{0.019} & 0.633\,\sd{0.003} & 0.271\,\sd{0.030} & 0.131\,\sd{0.021} & 0.194\,\sd{0.010} \\
\quad 4 & 0.664\,\sd{0.036} & 0.432\,\sd{0.056} & 0.635\,\sd{0.003} & 0.297\,\sd{0.036} & 0.123\,\sd{0.025} & 0.242\,\sd{0.012} \\
\cmidrule(l){1-7}
\quad \textit{mean} & 0.714\,\sd{0.013} & 0.507\,\sd{0.016} & 0.659\,\sd{0.003} & 0.312\,\sd{0.007} & 0.172\,\sd{0.001} & 0.245\,\sd{0.004} \\
\addlinespace[2pt]\midrule
\multicolumn{7}{l}{\textit{SageTCR}} \\
\quad 0 & 0.486\,\sd{0.013} & 0.260\,\sd{0.035} & 0.501\,\sd{0.056} & 0.150\,\sd{0.020} & 0.000\,\sd{0.000} & 0.182\,\sd{0.132} \\
\quad 1 & 0.634\,\sd{0.014} & 0.458\,\sd{0.031} & 0.512\,\sd{0.031} & 0.243\,\sd{0.007} & 0.163\,\sd{0.020} & 0.175\,\sd{0.010} \\
\quad 2 & 0.651\,\sd{0.033} & 0.487\,\sd{0.021} & 0.541\,\sd{0.011} & 0.237\,\sd{0.016} & 0.143\,\sd{0.038} & 0.165\,\sd{0.027} \\
\quad 3 & 0.680\,\sd{0.035} & 0.519\,\sd{0.037} & 0.568\,\sd{0.010} & 0.231\,\sd{0.010} & 0.129\,\sd{0.021} & 0.148\,\sd{0.004} \\
\quad 4 & 0.639\,\sd{0.023} & 0.460\,\sd{0.011} & 0.604\,\sd{0.061} & 0.274\,\sd{0.028} & 0.141\,\sd{0.028} & 0.202\,\sd{0.006} \\
\cmidrule(l){1-7}
\quad \textit{mean} & 0.618\,\sd{0.002} & 0.437\,\sd{0.014} & 0.545\,\sd{0.015} & 0.227\,\sd{0.003} & 0.115\,\sd{0.016} & 0.174\,\sd{0.030} \\
\addlinespace[2pt]\midrule
\multicolumn{7}{l}{\textit{Abeer}} \\
\quad 0 & 0.484\,\sd{0.024} & 0.329\,\sd{0.013} & 0.723\,\sd{0.043} & 0.388\,\sd{0.152} & 0.303\,\sd{0.196} & 0.344\,\sd{0.150} \\
\quad 1 & 0.614\,\sd{0.014} & 0.460\,\sd{0.032} & 0.514\,\sd{0.025} & 0.220\,\sd{0.010} & 0.009\,\sd{0.008} & 0.144\,\sd{0.006} \\
\quad 2 & 0.609\,\sd{0.026} & 0.438\,\sd{0.029} & 0.544\,\sd{0.001} & 0.215\,\sd{0.005} & 0.073\,\sd{0.024} & 0.132\,\sd{0.008} \\
\quad 3 & 0.636\,\sd{0.025} & 0.486\,\sd{0.022} & 0.546\,\sd{0.005} & 0.200\,\sd{0.006} & 0.049\,\sd{0.012} & 0.125\,\sd{0.014} \\
\quad 4 & 0.609\,\sd{0.052} & 0.442\,\sd{0.047} & 0.567\,\sd{0.002} & 0.260\,\sd{0.027} & 0.106\,\sd{0.014} & 0.196\,\sd{0.033} \\
\cmidrule(l){1-7}
\quad \textit{mean} & 0.590\,\sd{0.003} & 0.431\,\sd{0.002} & 0.579\,\sd{0.008} & 0.257\,\sd{0.028} & 0.108\,\sd{0.049} & 0.188\,\sd{0.025} \\
\addlinespace[2pt]\midrule
\multicolumn{7}{l}{\textit{TRAP}} \\
\quad 0 & 0.592\,\sd{0.008} & 0.438\,\sd{0.017} & 0.500\,\sd{0.124} & 0.155\,\sd{0.032} & 0.000\,\sd{0.000} & 0.102\,\sd{0.090} \\
\quad 1 & 0.657\,\sd{0.031} & 0.515\,\sd{0.042} & 0.439\,\sd{0.009} & 0.174\,\sd{0.006} & 0.064\,\sd{0.008} & 0.096\,\sd{0.010} \\
\quad 2 & 0.680\,\sd{0.009} & 0.551\,\sd{0.015} & 0.536\,\sd{0.007} & 0.194\,\sd{0.005} & 0.030\,\sd{0.004} & 0.081\,\sd{0.023} \\
\quad 3 & 0.678\,\sd{0.025} & 0.557\,\sd{0.019} & 0.528\,\sd{0.011} & 0.187\,\sd{0.005} & 0.045\,\sd{0.020} & 0.098\,\sd{0.012} \\
\quad 4 & 0.639\,\sd{0.041} & 0.506\,\sd{0.054} & 0.484\,\sd{0.010} & 0.190\,\sd{0.021} & 0.064\,\sd{0.056} & 0.102\,\sd{0.062} \\
\cmidrule(l){1-7}
\quad \textit{mean} & 0.649\,\sd{0.006} & 0.514\,\sd{0.012} & 0.497\,\sd{0.025} & 0.180\,\sd{0.003} & 0.041\,\sd{0.017} & 0.096\,\sd{0.008} \\
\bottomrule
\end{tabular}
\end{table}

\begin{table}[H]
\centering
\footnotesize\setlength{\tabcolsep}{2.2pt}
\caption{\textbf{Per-fold results, structures from TCRmodel2}. Conventions as in \cref{tab:perfold-tfold}.}
\label{tab:perfold-tcrmodel2}
\begin{tabular}{lcccccc}
\toprule
method / fold & macro-AUROC & macro-AUPRC & micro-AUROC & micro-AUPRC & MCC@0.5 & MCC$^{\ast}$ \\
\midrule
\multicolumn{7}{l}{\textit{\syb}} \\
\quad 0 & 0.729\,\sd{0.010} & 0.540\,\sd{0.006} & 0.856\,\sd{0.009} & 0.665\,\sd{0.025} & 0.623\,\sd{0.034} & 0.681\,\sd{0.021} \\
\quad 1 & 0.799\,\sd{0.018} & 0.639\,\sd{0.027} & 0.647\,\sd{0.003} & 0.283\,\sd{0.014} & 0.134\,\sd{0.027} & 0.175\,\sd{0.002} \\
\quad 2 & 0.794\,\sd{0.006} & 0.605\,\sd{0.018} & 0.568\,\sd{0.007} & 0.261\,\sd{0.017} & 0.172\,\sd{0.014} & 0.185\,\sd{0.020} \\
\quad 3 & 0.775\,\sd{0.022} & 0.626\,\sd{0.021} & 0.620\,\sd{0.002} & 0.275\,\sd{0.023} & 0.149\,\sd{0.045} & 0.179\,\sd{0.015} \\
\quad 4 & 0.747\,\sd{0.009} & 0.554\,\sd{0.018} & 0.648\,\sd{0.015} & 0.351\,\sd{0.017} & 0.256\,\sd{0.047} & 0.284\,\sd{0.009} \\
\cmidrule(l){1-7}
\quad \textit{mean} & 0.769\,\sd{0.007} & 0.593\,\sd{0.004} & 0.668\,\sd{0.003} & 0.367\,\sd{0.002} & 0.267\,\sd{0.009} & 0.301\,\sd{0.003} \\
\addlinespace[2pt]\midrule
\multicolumn{7}{l}{\textit{STAG-LLM}} \\
\quad 0 & 0.599\,\sd{0.042} & 0.457\,\sd{0.028} & 0.794\,\sd{0.020} & 0.481\,\sd{0.068} & 0.372\,\sd{0.110} & 0.492\,\sd{0.080} \\
\quad 1 & 0.640\,\sd{0.051} & 0.469\,\sd{0.035} & 0.478\,\sd{0.014} & 0.226\,\sd{0.026} & 0.060\,\sd{0.007} & 0.173\,\sd{0.026} \\
\quad 2 & 0.673\,\sd{0.016} & 0.518\,\sd{0.004} & 0.535\,\sd{0.005} & 0.224\,\sd{0.004} & 0.060\,\sd{0.035} & 0.170\,\sd{0.010} \\
\quad 3 & 0.669\,\sd{0.021} & 0.512\,\sd{0.022} & 0.569\,\sd{0.010} & 0.271\,\sd{0.018} & 0.078\,\sd{0.024} & 0.222\,\sd{0.020} \\
\quad 4 & 0.577\,\sd{0.030} & 0.427\,\sd{0.019} & 0.553\,\sd{0.013} & 0.272\,\sd{0.023} & 0.196\,\sd{0.018} & 0.262\,\sd{0.017} \\
\cmidrule(l){1-7}
\quad \textit{mean} & 0.632\,\sd{0.009} & 0.477\,\sd{0.005} & 0.586\,\sd{0.008} & 0.295\,\sd{0.018} & 0.153\,\sd{0.029} & 0.264\,\sd{0.019} \\
\addlinespace[2pt]\midrule
\multicolumn{7}{l}{\textit{SageTCR}} \\
\quad 0 & 0.470\,\sd{0.035} & 0.187\,\sd{0.026} & 0.465\,\sd{0.048} & 0.145\,\sd{0.009} & 0.000\,\sd{0.000} & 0.112\,\sd{0.087} \\
\quad 1 & 0.624\,\sd{0.036} & 0.473\,\sd{0.023} & 0.522\,\sd{0.013} & 0.249\,\sd{0.013} & 0.105\,\sd{0.045} & 0.176\,\sd{0.023} \\
\quad 2 & 0.617\,\sd{0.048} & 0.465\,\sd{0.049} & 0.548\,\sd{0.007} & 0.241\,\sd{0.027} & 0.127\,\sd{0.053} & 0.167\,\sd{0.034} \\
\quad 3 & 0.677\,\sd{0.017} & 0.524\,\sd{0.019} & 0.572\,\sd{0.007} & 0.249\,\sd{0.010} & 0.138\,\sd{0.009} & 0.174\,\sd{0.007} \\
\quad 4 & 0.593\,\sd{0.035} & 0.439\,\sd{0.016} & 0.578\,\sd{0.001} & 0.292\,\sd{0.010} & 0.190\,\sd{0.020} & 0.238\,\sd{0.026} \\
\cmidrule(l){1-7}
\quad \textit{mean} & 0.596\,\sd{0.028} & 0.418\,\sd{0.011} & 0.537\,\sd{0.012} & 0.235\,\sd{0.004} & 0.112\,\sd{0.018} & 0.174\,\sd{0.011} \\
\addlinespace[2pt]\midrule
\multicolumn{7}{l}{\textit{Abeer}} \\
\quad 0 & 0.486\,\sd{0.065} & 0.348\,\sd{0.056} & 0.780\,\sd{0.011} & 0.519\,\sd{0.098} & 0.372\,\sd{0.177} & 0.517\,\sd{0.004} \\
\quad 1 & 0.600\,\sd{0.029} & 0.443\,\sd{0.028} & 0.531\,\sd{0.021} & 0.226\,\sd{0.011} & 0.088\,\sd{0.076} & 0.144\,\sd{0.020} \\
\quad 2 & 0.633\,\sd{0.025} & 0.476\,\sd{0.016} & 0.542\,\sd{0.008} & 0.214\,\sd{0.005} & 0.072\,\sd{0.028} & 0.137\,\sd{0.022} \\
\quad 3 & 0.618\,\sd{0.026} & 0.485\,\sd{0.025} & 0.563\,\sd{0.015} & 0.215\,\sd{0.017} & 0.062\,\sd{0.034} & 0.142\,\sd{0.051} \\
\quad 4 & 0.580\,\sd{0.012} & 0.421\,\sd{0.023} & 0.570\,\sd{0.007} & 0.280\,\sd{0.024} & 0.132\,\sd{0.019} & 0.240\,\sd{0.040} \\
\cmidrule(l){1-7}
\quad \textit{mean} & 0.583\,\sd{0.023} & 0.434\,\sd{0.017} & 0.597\,\sd{0.006} & 0.291\,\sd{0.021} & 0.145\,\sd{0.024} & 0.236\,\sd{0.019} \\
\addlinespace[2pt]\midrule
\multicolumn{7}{l}{\textit{TRAP}} \\
\quad 0 & 0.578\,\sd{0.028} & 0.409\,\sd{0.032} & 0.460\,\sd{0.069} & 0.147\,\sd{0.017} & 0.000\,\sd{0.000} & 0.058\,\sd{0.031} \\
\quad 1 & 0.635\,\sd{0.068} & 0.498\,\sd{0.099} & 0.430\,\sd{0.056} & 0.159\,\sd{0.021} & 0.003\,\sd{0.013} & 0.068\,\sd{0.039} \\
\quad 2 & 0.692\,\sd{0.022} & 0.573\,\sd{0.012} & 0.521\,\sd{0.011} & 0.197\,\sd{0.006} & 0.045\,\sd{0.007} & 0.104\,\sd{0.013} \\
\quad 3 & 0.671\,\sd{0.018} & 0.549\,\sd{0.016} & 0.526\,\sd{0.002} & 0.191\,\sd{0.003} & 0.046\,\sd{0.018} & 0.111\,\sd{0.010} \\
\quad 4 & 0.572\,\sd{0.019} & 0.401\,\sd{0.024} & 0.496\,\sd{0.025} & 0.167\,\sd{0.016} & 0.000\,\sd{0.000} & 0.043\,\sd{0.026} \\
\cmidrule(l){1-7}
\quad \textit{mean} & 0.630\,\sd{0.008} & 0.486\,\sd{0.011} & 0.487\,\sd{0.009} & 0.172\,\sd{0.004} & 0.019\,\sd{0.006} & 0.077\,\sd{0.006} \\
\bottomrule
\end{tabular}
\end{table}

\begin{table}[H]
\centering
\footnotesize\setlength{\tabcolsep}{2.2pt}
\caption{\textbf{Per-fold results, structures from Boltz-2}. Conventions as in \cref{tab:perfold-tfold}.}
\label{tab:perfold-boltz}
\begin{tabular}{lcccccc}
\toprule
method / fold & macro-AUROC & macro-AUPRC & micro-AUROC & micro-AUPRC & MCC@0.5 & MCC$^{\ast}$ \\
\midrule
\multicolumn{7}{l}{\textit{\syb}} \\
\quad 0 & 0.748\,\sd{0.014} & 0.572\,\sd{0.021} & 0.861\,\sd{0.006} & 0.671\,\sd{0.021} & 0.627\,\sd{0.022} & 0.666\,\sd{0.034} \\
\quad 1 & 0.770\,\sd{0.013} & 0.598\,\sd{0.022} & 0.628\,\sd{0.039} & 0.316\,\sd{0.013} & 0.173\,\sd{0.033} & 0.225\,\sd{0.001} \\
\quad 2 & 0.770\,\sd{0.010} & 0.586\,\sd{0.010} & 0.591\,\sd{0.009} & 0.285\,\sd{0.012} & 0.182\,\sd{0.040} & 0.219\,\sd{0.021} \\
\quad 3 & 0.751\,\sd{0.015} & 0.598\,\sd{0.012} & 0.603\,\sd{0.008} & 0.275\,\sd{0.021} & 0.173\,\sd{0.028} & 0.189\,\sd{0.027} \\
\quad 4 & 0.727\,\sd{0.011} & 0.538\,\sd{0.006} & 0.623\,\sd{0.006} & 0.311\,\sd{0.006} & 0.236\,\sd{0.028} & 0.256\,\sd{0.036} \\
\cmidrule(l){1-7}
\quad \textit{mean} & 0.753\,\sd{0.004} & 0.579\,\sd{0.007} & 0.661\,\sd{0.006} & 0.372\,\sd{0.003} & 0.278\,\sd{0.013} & 0.311\,\sd{0.012} \\
\addlinespace[2pt]\midrule
\multicolumn{7}{l}{\textit{STAG-LLM}} \\
\quad 0 & 0.696\,\sd{0.007} & 0.483\,\sd{0.005} & 0.811\,\sd{0.028} & 0.460\,\sd{0.097} & 0.342\,\sd{0.022} & 0.488\,\sd{0.093} \\
\quad 1 & 0.741\,\sd{0.011} & 0.544\,\sd{0.012} & 0.613\,\sd{0.011} & 0.304\,\sd{0.016} & 0.135\,\sd{0.038} & 0.206\,\sd{0.014} \\
\quad 2 & 0.734\,\sd{0.007} & 0.525\,\sd{0.022} & 0.607\,\sd{0.007} & 0.243\,\sd{0.002} & 0.114\,\sd{0.012} & 0.141\,\sd{0.015} \\
\quad 3 & 0.718\,\sd{0.014} & 0.534\,\sd{0.005} & 0.633\,\sd{0.003} & 0.274\,\sd{0.020} & 0.121\,\sd{0.012} & 0.202\,\sd{0.005} \\
\quad 4 & 0.661\,\sd{0.023} & 0.441\,\sd{0.027} & 0.614\,\sd{0.017} & 0.290\,\sd{0.013} & 0.131\,\sd{0.014} & 0.231\,\sd{0.017} \\
\cmidrule(l){1-7}
\quad \textit{mean} & 0.710\,\sd{0.002} & 0.505\,\sd{0.003} & 0.655\,\sd{0.008} & 0.314\,\sd{0.026} & 0.169\,\sd{0.012} & 0.254\,\sd{0.020} \\
\addlinespace[2pt]\midrule
\multicolumn{7}{l}{\textit{SageTCR}} \\
\quad 0 & 0.482\,\sd{0.016} & 0.226\,\sd{0.064} & 0.452\,\sd{0.043} & 0.140\,\sd{0.015} & 0.000\,\sd{0.000} & 0.033\,\sd{0.017} \\
\quad 1 & 0.611\,\sd{0.010} & 0.457\,\sd{0.028} & 0.475\,\sd{0.029} & 0.229\,\sd{0.030} & 0.155\,\sd{0.054} & 0.172\,\sd{0.034} \\
\quad 2 & 0.659\,\sd{0.029} & 0.503\,\sd{0.024} & 0.554\,\sd{0.003} & 0.241\,\sd{0.004} & 0.134\,\sd{0.038} & 0.171\,\sd{0.011} \\
\quad 3 & 0.644\,\sd{0.035} & 0.508\,\sd{0.024} & 0.578\,\sd{0.009} & 0.237\,\sd{0.009} & 0.109\,\sd{0.007} & 0.148\,\sd{0.014} \\
\quad 4 & 0.591\,\sd{0.019} & 0.416\,\sd{0.017} & 0.575\,\sd{0.023} & 0.272\,\sd{0.014} & 0.177\,\sd{0.008} & 0.205\,\sd{0.013} \\
\cmidrule(l){1-7}
\quad \textit{mean} & 0.597\,\sd{0.001} & 0.422\,\sd{0.013} & 0.527\,\sd{0.012} & 0.224\,\sd{0.009} & 0.115\,\sd{0.003} & 0.146\,\sd{0.004} \\
\addlinespace[2pt]\midrule
\multicolumn{7}{l}{\textit{Abeer}} \\
\quad 0 & 0.516\,\sd{0.026} & 0.357\,\sd{0.021} & 0.753\,\sd{0.033} & 0.409\,\sd{0.148} & 0.086\,\sd{0.153} & 0.387\,\sd{0.123} \\
\quad 1 & 0.621\,\sd{0.017} & 0.460\,\sd{0.015} & 0.488\,\sd{0.017} & 0.217\,\sd{0.014} & 0.010\,\sd{0.079} & 0.150\,\sd{0.027} \\
\quad 2 & 0.648\,\sd{0.014} & 0.472\,\sd{0.012} & 0.538\,\sd{0.018} & 0.202\,\sd{0.014} & 0.049\,\sd{0.045} & 0.101\,\sd{0.027} \\
\quad 3 & 0.624\,\sd{0.054} & 0.465\,\sd{0.048} & 0.562\,\sd{0.035} & 0.208\,\sd{0.029} & 0.065\,\sd{0.049} & 0.129\,\sd{0.040} \\
\quad 4 & 0.594\,\sd{0.017} & 0.432\,\sd{0.018} & 0.558\,\sd{0.013} & 0.238\,\sd{0.006} & 0.062\,\sd{0.053} & 0.172\,\sd{0.019} \\
\cmidrule(l){1-7}
\quad \textit{mean} & 0.600\,\sd{0.003} & 0.437\,\sd{0.009} & 0.580\,\sd{0.014} & 0.255\,\sd{0.032} & 0.054\,\sd{0.026} & 0.188\,\sd{0.023} \\
\addlinespace[2pt]\midrule
\multicolumn{7}{l}{\textit{TRAP}} \\
\quad 0 & 0.552\,\sd{0.038} & 0.395\,\sd{0.039} & 0.461\,\sd{0.071} & 0.146\,\sd{0.019} & 0.000\,\sd{0.000} & 0.054\,\sd{0.034} \\
\quad 1 & 0.686\,\sd{0.029} & 0.558\,\sd{0.027} & 0.472\,\sd{0.031} & 0.185\,\sd{0.008} & 0.009\,\sd{0.028} & 0.109\,\sd{0.009} \\
\quad 2 & 0.681\,\sd{0.008} & 0.540\,\sd{0.019} & 0.521\,\sd{0.002} & 0.188\,\sd{0.001} & 0.019\,\sd{0.008} & 0.090\,\sd{0.006} \\
\quad 3 & 0.672\,\sd{0.075} & 0.544\,\sd{0.083} & 0.501\,\sd{0.031} & 0.169\,\sd{0.022} & 0.024\,\sd{0.022} & 0.075\,\sd{0.054} \\
\quad 4 & 0.622\,\sd{0.053} & 0.473\,\sd{0.088} & 0.502\,\sd{0.002} & 0.186\,\sd{0.015} & 0.030\,\sd{0.053} & 0.084\,\sd{0.042} \\
\cmidrule(l){1-7}
\quad \textit{mean} & 0.642\,\sd{0.015} & 0.502\,\sd{0.006} & 0.491\,\sd{0.011} & 0.175\,\sd{0.002} & 0.016\,\sd{0.013} & 0.082\,\sd{0.009} \\
\bottomrule
\end{tabular}
\end{table}

\begin{table}[H]
\centering
\footnotesize\setlength{\tabcolsep}{2.2pt}
\caption{\textbf{Per-fold results, sequence baselines}. These methods read no structure, so they are listed once. Each cell is the mean over three training seeds, $\pm$ the spread across those seeds.}
\label{tab:perfold-seq}
\begin{tabular}{lcccccc}
\toprule
method / fold & macro-AUROC & macro-AUPRC & micro-AUROC & micro-AUPRC & MCC@0.5 & MCC$^{\ast}$ \\
\midrule
\multicolumn{7}{l}{\textit{NetTCR-2.2}} \\
\quad 0 & 0.676\,\sd{0.011} & 0.530\,\sd{0.006} & 0.806\,\sd{0.021} & 0.541\,\sd{0.063} & 0.453\,\sd{0.128} & 0.519\,\sd{0.083} \\
\quad 1 & 0.699\,\sd{0.024} & 0.519\,\sd{0.029} & 0.478\,\sd{0.009} & 0.230\,\sd{0.008} & 0.093\,\sd{0.005} & 0.161\,\sd{0.008} \\
\quad 2 & 0.702\,\sd{0.003} & 0.550\,\sd{0.007} & 0.541\,\sd{0.003} & 0.231\,\sd{0.008} & 0.120\,\sd{0.027} & 0.162\,\sd{0.018} \\
\quad 3 & 0.689\,\sd{0.024} & 0.549\,\sd{0.036} & 0.565\,\sd{0.006} & 0.222\,\sd{0.009} & 0.093\,\sd{0.018} & 0.145\,\sd{0.038} \\
\quad 4 & 0.676\,\sd{0.021} & 0.530\,\sd{0.013} & 0.631\,\sd{0.047} & 0.370\,\sd{0.016} & 0.276\,\sd{0.008} & 0.299\,\sd{0.014} \\
\cmidrule(l){1-7}
\quad \textit{mean} & 0.688\,\sd{0.007} & 0.536\,\sd{0.012} & 0.604\,\sd{0.010} & 0.319\,\sd{0.009} & 0.207\,\sd{0.026} & 0.257\,\sd{0.006} \\
\addlinespace[2pt]\midrule
\multicolumn{7}{l}{\textit{ERGO-II}} \\
\quad 0 & 0.658\,\sd{0.034} & 0.489\,\sd{0.027} & 0.755\,\sd{0.041} & 0.428\,\sd{0.082} & 0.219\,\sd{0.085} & 0.383\,\sd{0.105} \\
\quad 1 & 0.702\,\sd{0.007} & 0.540\,\sd{0.013} & 0.482\,\sd{0.005} & 0.234\,\sd{0.011} & -0.070\,\sd{0.006} & 0.175\,\sd{0.021} \\
\quad 2 & 0.706\,\sd{0.022} & 0.541\,\sd{0.033} & 0.555\,\sd{0.005} & 0.246\,\sd{0.015} & 0.039\,\sd{0.007} & 0.173\,\sd{0.015} \\
\quad 3 & 0.645\,\sd{0.020} & 0.480\,\sd{0.020} & 0.572\,\sd{0.003} & 0.229\,\sd{0.014} & 0.060\,\sd{0.003} & 0.165\,\sd{0.012} \\
\quad 4 & 0.631\,\sd{0.030} & 0.462\,\sd{0.027} & 0.599\,\sd{0.043} & 0.285\,\sd{0.045} & 0.054\,\sd{0.043} & 0.222\,\sd{0.061} \\
\cmidrule(l){1-7}
\quad \textit{mean} & 0.668\,\sd{0.016} & 0.502\,\sd{0.016} & 0.593\,\sd{0.009} & 0.284\,\sd{0.010} & 0.061\,\sd{0.026} & 0.224\,\sd{0.012} \\
\bottomrule
\end{tabular}
\end{table}

\subsection{Architecture Ablations}
\label{app:ablations}
\vspace{0.1in}

Two more ablations were run on the
same protocol and are reported here. Removing the graph but keeping the structure-derived node descriptors costs 13.1\% of the
micro-AUPRC margin the control holds above chance, which is close to what removing the structure
entirely costs, so most of what the structure contributes arrives through the neighborhood rather
than through the per-residue descriptors. Restoring the TCR $\alpha$ chain costs 8.6\% of that
margin and 10.0\% of MCC@0.5, which is the measurement behind the scope decision of
\cref{app:scope}.

\textbf{The docked-complex arm.} \emph{docked complex} replaces the two monomer structures with the
ESMFold2 co-folded complex of the same pair. The two halves already share one frame. The arm
concatenates them into one node set and runs a single encoder pass. The control runs one pass per
molecule. Cross-molecular edges appear on their own, because the $k$-NN graph is rebuilt on the
union at every layer. No cross-edge flag and no new feature is added. The five biochemical edge
terms read residue identity alone, so they were already defined across the boundary. The contact
head is unchanged. The complex cache keeps the node order of the monomer graph, so the arm reuses
the same language-model cache and the same CDR3$\beta$ and peptide row sets. Capacity is unchanged
as well. Both arms share one encoder between the two molecules, and the docked arm runs that same
encoder over the union. The two arms hold $3{,}027{,}721$ parameters under identical names and
shapes. Node features are recomputed from the complex instead of copied from the monomer graphs,
because a co-folded half differs from its separately folded monomer. Cross-molecular edges carry
distance and residue identity only, so the score stays invariant to a rigid motion of the pair. It
is not invariant to independent motions of the two halves. That second invariance is what
\cref{prop:pose} states and what this arm gives up.

\begin{table}[H]
\centering
\footnotesize\setlength{\tabcolsep}{4pt}
\caption{\textbf{The full ablation ladder, on the ESMFold2 structures.} The first four rows are
\cref{tab:ablation}, repeated so the two remaining arms can be read against them. \emph{no graph}
replaces the encoder with a simple multilayer perceptron, which removes the neighborhood and the coordinate
update while keeping the structure-derived node descriptors. \emph{alpha chain restored} puts the TCR $\alpha$ chain back into the
encoder input (\cref{app:scope}). Means over five folds and three seeds,
$\pm$ the spread across seeds.}
\label{tab:ablation-full}
\resizebox{\textwidth}{!}{%
\begin{tabular}{lcccccc}
\toprule
arm & macro-AUROC & macro-AUPRC & micro-AUROC & micro-AUPRC & MCC@0.5 & MCC$^{\ast}$ \\
\midrule
\sys (control) & 0.759\,\sd{0.004} & 0.584\,\sd{0.003} & 0.670\,\sd{0.005} & 0.372\,\sd{0.013} & 0.273\,\sd{0.029} & 0.304\,\sd{0.011} \\
\quad remove structure & 0.753\,\sd{0.008} & 0.570\,\sd{0.006} & 0.649\,\sd{0.003} & 0.349\,\sd{0.001} & 0.236\,\sd{0.018} & 0.283\,\sd{0.010} \\
\quad docked complex & 0.745\,\sd{0.009} & 0.578\,\sd{0.013} & 0.654\,\sd{0.009} & 0.346\,\sd{0.033} & 0.232\,\sd{0.063} & 0.274\,\sd{0.027} \\
\quad no peptide frame & 0.740\,\sd{0.008} & 0.558\,\sd{0.004} & 0.664\,\sd{0.007} & 0.360\,\sd{0.004} & 0.209\,\sd{0.017} & 0.279\,\sd{0.009} \\
\midrule
\quad no graph & 0.750\,\sd{0.002} & 0.561\,\sd{0.002} & 0.656\,\sd{0.006} & 0.344\,\sd{0.007} & 0.249\,\sd{0.012} & 0.275\,\sd{0.020} \\
\quad alpha chain restored & 0.739\,\sd{0.020} & 0.568\,\sd{0.026} & 0.643\,\sd{0.004} & 0.354\,\sd{0.014} & 0.246\,\sd{0.025} & 0.280\,\sd{0.012} \\
\bottomrule
\end{tabular}%
}
\end{table}

\subsection{Preprocessing Benchmark}
\label{app:cost}
\vspace{0.1in}

\textbf{What is timed.} The sample behind \cref{tab:cost} is fixed once and shared by all four
tools, because folding time scales with residue count. Residue medians are 228 for a receptor, 382
for a target and 610 for a complex. Each cell is GPU time for the inference call alone. Model load
is excluded and timed separately. The first structure of each run is excluded as well, since it pays
one-time kernel autotuning.

\textbf{Two rows carry a cost the table does not show.} TCRmodel2 spends 45 to 48 percent of its
wall clock outside inference, on container start, imports, parameter loading, template search and
featurization, so its $2.6\times$ is a ratio between two numbers that both carry that overhead.
Boltz-2 cannot fold without a multiple sequence alignment, and that alignment is not in its row.
Every chain carries an explicit alignment path and the 9-mer peptide carries none, which matches the
production configuration rather than the cheaper one. Timing it against a public alignment server
would have compared a different quantity to tFold-TCR and ESMFold2, which need no alignment at all.

\textbf{Limits.} Folding a complex is not expensive in itself, and on a fixed candidate set the
advantage is a constant factor rather than a change of regime. Each cell is one measurement rather
than a mean over repeats, so no error bar is claimed, though the within-cell spread is tight at
$p_{90}$ over median of 1.03 to 1.08.

\end{document}